\documentclass{elsarticle}

\usepackage{graphicx}%
\usepackage{multirow}%
\usepackage{amsmath,amssymb,amsfonts}%
\usepackage{amsthm}%
\usepackage{mathrsfs}%
\usepackage[title]{appendix}%
\usepackage{xcolor}%
\usepackage{textcomp}%
\usepackage{manyfoot}%
\usepackage{booktabs}%
\usepackage{algorithm}%
\usepackage{algorithmicx}%
\usepackage{algpseudocode}%
\usepackage{listings}%
\usepackage[T1]{fontenc}
\usepackage{graphicx}
\usepackage{bussproofs}
\usepackage{xspace}
\usepackage{tikz}
\usepackage{cancel}
\usepackage{caption}
\usepackage{subcaption}
\usepackage{gensymb}
\usepackage{comment}
\usepackage{url}
\usepackage{hyperref}
\usepackage{cleveref}%
\usepackage{adjustbox}

\newtheorem{example}{Example}
\newtheorem{definition}{Definition}
\newtheorem{theorem}{Theorem}

\newtheorem{lemma}{Lemma}

\newcommand{\qlang}{\ensuremath{\mathtt{QLang}}\xspace}
\newcommand{\experiments}{\ensuremath{E}\xspace}
\newcommand{\erlang}{\ensuremath{\mathtt{RLang}}\xspace}
\newcommand{\exscheme}{\ensuremath{\mathtt{ELang}}\xspace}
\newcommand{\quest}{\ensuremath{\mathsf{que}}\xspace}
\newcommand{\answer}{\ensuremath{\mathsf{ans}}\xspace}

\newcommand{\request}{\ensuremath{\mathsf{req}}\xspace}
\newcommand{\response}{\ensuremath{\mathsf{rsp}}\xspace}
\newcommand{\decomp}{\ensuremath{\mathsf{decompose}}\xspace}
\newcommand{\aggregate}{\ensuremath{\mathsf{aggregate}}\xspace}
\newcommand{\experiment}{\ensuremath{\mathsf{exp}}\xspace}
\newcommand{\exresult}{\ensuremath{\mathsf{res}}\xspace}
\newcommand{\execute}{\ensuremath{\mathsf{execute}}\xspace}
\newcommand{\complete}{\ensuremath{\mathsf{complete}}\xspace}
\newcommand{\compute}{\ensuremath{\mathsf{compute}}\xspace}
\newcommand{\fulfilled}{\ensuremath{\mathsf{fulfilledBy}}\xspace}
\newcommand{\answered}{\ensuremath{\mathsf{answeredBy}}\xspace}
\newcommand{\dist}[1]{\left|#1\right|}

\newcommand{\justify}{\ensuremath{\mathtt{justify}}\xspace}

\colorlet{keywordcolor}{blue!50!black}
\colorlet{commentcolor}{green!60!black}
\colorlet{typecolor}{green!70!black}

\lstdefinelanguage{pseudo}{
        basicstyle=\ttfamily\footnotesize\upshape,
        keywordstyle=\color{keywordcolor}\bfseries\sffamily,
        commentstyle=\slshape\rmfamily\color{commentcolor},
        columns=fullflexible,
        mathescape=true,
        escapechar={\#},
        keepspaces=true,
        showstringspaces=false,
        aboveskip=8pt, 
        numbers=left,
        stepnumber=1, 
        numberstyle=\ttfamily\scriptsize\color{gray},
        numbersep=4pt,
        xleftmargin=1.5em,
        xrightmargin=1.5em,
        framexleftmargin=1.2em,
        framexrightmargin=1em,
        framextopmargin=0.5ex,
        breaklines=true,
        breakindent=3pt,
        keywords={function,Input:,Output:,do, if,return,end,else,foreach},keywordstyle=\color{keywordcolor}\bfseries\sffamily,
        morekeywords=[2]{link, load, anchor, retrieve, extends},
        keywordstyle=[2]\color{typecolor},
        sensitive=true,
        comment=[l]{//},
        morecomment=[s]{/*}{*/},
        morestring=[b]"
        }

\begin{document}
\begin{frontmatter}
\title{Reasonable Experiments in Model-Based~Systems~Engineering}

\author{Johan Cederbladh}
\ead{johan.cederbladh@mdu.se}
\address{Mälardalen University, Sweden}
\author{Loek Cleophas}
\ead{l.g.w.a.cleophas@tue.nl}
\address{Eindhoven University of Technology, the Netherlands,\\ and Stellenbosch University, South Africa}
\author{Eduard Kamburjan}
\ead{eduard.kamburjan@itu.dk}
\address{IT University of Copenhagen, Denmark,\\ and University of Oslo, Norway}
\author{Lucas Lima}
\ead{lucas.albertins@ufrpe.br}
\address{Universidade Federal Rural de Pernambuco, Brazil}
\author{Rakshit Mittal}
\ead{rakshit.mittal@uantwerpen.be}
\author{Hans Vangheluwe}
\ead{hans.vangheluwe@uantwerpen.be}
\address{University of Antwerp - Flanders Make, Belgium}
\begin{abstract}
With the current trend in Model-Based Systems Engineering towards Digital Engineering and early Validation \& Verification, \textit{experiments} are increasingly used to estimate system parameters and explore design decisions. Managing such experimental configuration metadata and results is of utmost importance in accelerating overall design effort. In particular, we observe it is important to 'intelligent-ly' reuse experiment-related data to save time and effort by not performing potentially superfluous, time-consuming, and resource-intensive experiments.
In this work, we present a framework for managing experiments on digital and/or physical assets with a focus on case-based reasoning with domain knowledge to reuse experimental data efficiently by deciding whether an already-performed experiment (or associated answer) can be reused to answer a new (potentially different) question from the engineer/user without having to set up and perform a new experiment. We provide the general architecture for such an experiment manager and validate our approach using an industrial vehicular energy system-design case study.
\end{abstract}

\end{frontmatter}

\section{Introduction}\label{sec:intro}

Systems Engineering (SE) at its core revolves around \textit{risk} management~\cite{walden2023systems,haberfellner2019systems}. Carefully crafted design procedures are employed throughout the development process to reduce uncertainty until a definitive design decision can be made~\cite{cederbladh2023early}. Early design stages are particularly critical to reduce uncertainty and subsequent development effort. In this context, Model-Based Systems Engineering (MBSE) is recognized as the \textit{future of SE}~\cite{madni2018model,henderson2021value} and enables earlier Verification and Validation (V\&V) of system design through the early and continuous availability of models~\cite{cederbladh2023early}. Models enable digital experiments and management of analytical results, and experiment data gathered in previous endeavors. One aim is to reduce workload and cost through \emph{reuse} of experiments. 

Both virtual experiments with models~\cite{puntigam2020integrated,abdo2022model} in MBSE and physical experiments in classical SE~\cite{walden2023systems,fisher1966design} are highly complex activities that require  elaborate structures enabling users to perform the right experiments for their queries. The optimal reuse of experimental data requires (a) careful data and model management~\cite{10.1145/3652620.3688221, 10.1145/3652620.3688223}, and (b) deep domain knowledge to conclude whether a new situation requires new experiments or can be reasoned over using existing results. However, frameworks realizing these features are often developed ad-hoc at companies and are not systematically deployed. This hampers the adoption of early V\&V in SE~\cite{suryadevara2018adopting,gregory2020long,gustavsson2024success}, especially in MBSE which requires additional knowledge about models: lack of explicitly represented knowledge about models and the domain are the critical bottlenecks that hamper reuse efforts in MBSE.

In our previous work \cite{10350626} we briefly introduced the general framework for reuse of experiments, without a formal description or evaluation on a reference implementation. In this paper, we build on this previous work, and address the gap of a common framework to enable efficient early decision-making through reuse of experiments, using domain knowledge. 
To this end, we (a) formalize experiment-related structures, ranging from high-level user queries to low-level experiment specification, (b) describe how reuse is enabled through different case-based reasoning techniques over these structures, (c) devise a software architecture to implement experiment structures and their reuse, and (d) evaluate our framework with an industrial case-study. In contrast to prior model and experiment modeling approaches, e.g., those for validity frames~\cite{AckerMVD24,mittal23},  planning of experiments or explainability, we focus on modeling aspects of experiments related to reuse i.e. it is assumed that the models have been validated prior to their reuse.


\paragraph{Experiment Structures}
Experiment management comprises three different layers as shown in \Cref{fig:intro-framework}: A \emph{user layer} exposes the system to the user and accepts user queries, a \emph{decomposition layer} decomposes a user query into a set of experiment requests, and an \emph{execution layer} schedules, executes, and records experiments.

The user layer has the capability to expose multiple languages conforming to queries from different user perspectives. 
For example, consider a company designing trains. Non-technical users may query whether a given train configuration is safe according to some regulation, which may involve a number of experiments that they do not have the knowledge to set up. Technical users ask more specific queries, e.g., whether \textit{``a particular train configuration $x$ stops after $y$ meters in a particular scenario $z$?''}. 

A query is a more abstract construct than an experiment request and an experiment. Internally, a user query corresponds to one or more experiment \emph{requests}. The decomposition layer encodes the reasoning used to decompose such user queries into experiment requests. In contrast to a query, a request is a technical perspective on the experimental activity and not the system-under-study (views on the system-under-study are handled by the user layer). Experiment requests are higher-level (than experiment specifications) descriptions of experiments \emph{and their relation to the user query}. Each request is fulfilled by exactly one response.

Finally, an experiment must be fully specified and performed by the execution layer -- \emph{experiment specifications} fully describe the experiment, including internal structure such as (physical) setup, enabling complete traceability including vertical, horizontal, and fine-grained traceabilities \cite{RYS2024100720}. Each experiment specification results in exactly one experiment \emph{result} that is returned from the execution layer. Since an \emph{experiment specification} is highly domain-specific, we remain agnostic to any \emph{experiment specification} scheme in our framework, while tracing the requirements that such a specification needs to satisfy in order to be compatible with our framework. 

\begin{figure}[h]
\centering
\includegraphics[width=\linewidth]{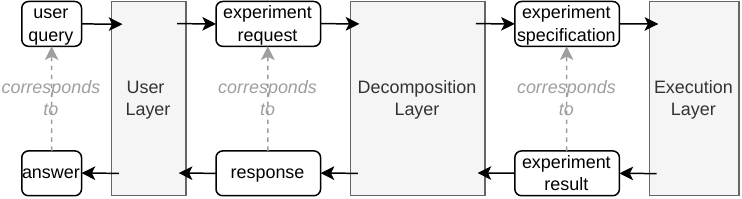}
\caption{High-level view of the proposed framework}
\label{fig:intro-framework}
\end{figure}

\paragraph{Reuse through Reasoning}
Reuse is possible within each of the above described layers and the distinction between user query and experiment requests/specifications enables to implement reuse on multiple levels:
\begin{itemize}
\item caching (by memoization) and reusing answers to user queries,
\item domain knowledge about the systems, models and experiments, to reuse responses to experiment requests, and
\item an experiment store to reuse experiment results from experiment specifications.
\end{itemize}
Nonetheless, the underlying mechanisms are similar for all layers. 

While simple reuse is realized by retrieving the answer to an already posed query or request, we also employ \emph{reasoning-based reuse considering explicitly represented knowledge}. Responses, experiment data, and the posed queries and requests are accumulated in a \emph{data store} equipped with \emph{case-based reasoning} (CBR)~\cite{cbr} capabilities, which attempt to derive an answer from a previously reached answer without delegating to the next lower-level layer. 
\begin{itemize} 
\item A naive example of CBR -- if the train does not stop in a certain setup within $x$ meters, then it will surely also not stop within $y < x$ meters in the same setup. The first answer thus justifies an answer to the second, without the need to perform a new experiment. 
\item A more complex example -- if the train is found to stop in $x$ meters in an experiment with $a$ initial velocity, then the train will surely stop within $y > x$ meters with $b < a$ initial velocity (since the stopping distance is proportional to the initial velocity, see \Cref{sec:example}) and everything else in the experiment configuration being identical. Thus the first answer can be used to infer an answer for the second query, without the need to perform a new experiment.
\end{itemize}

As CBR is limited to symbolic and precise reasoning, but in early V\&V a certain level of imprecision is acceptable, we also devise \emph{fuzzy retrieval} through similarity search~\cite{DBLP:series/ads/ZezulaADB06}.
We do so in two ways:
\begin{itemize}
\item Within a given level of uncertainty within early SE, when a query is considered similar enough to a previous one, then the previous answer can be reused directly. This answer is not the exact one to the posed query, but has a bounded distance to the (implicit) correct answer. For example, if the train does not stop within 400$m$, then the query whether it stops within 401$m$ may be similar enough to not need a new exact experiment execution.
\item The second approach describes when a query is considered similar enough, such that the same result data can be used to compute a new answer. This is useful, for example, for a changed property of interest. For example, if there is an experiment for a train with a certain configuration, performed because of a query about its stopping distance, then a query about its energy consumption with the same configuration can reuse already recorded data.
\end{itemize}

We evaluate our proposed framework with an industrial energy systems design case study, namely architecture design for battery system layouts~\cite{Cederbladh6985}, to demonstrate that we can indeed reuse a substantial number of experiments.

\paragraph{Contributions and Structure}
Our main contributions are (a) a formalization of experiment reuse, (b) a mapping of this formalization to a software architecture, and (c) an evaluation (with reproducible artifacts) using an industrial energy system design case study. The formalization is illustrated with a running example of a highly-simplified academic train braking model and the simulation of its physics, as has been done in previous paragraphs. 

This work is structured as follows. \Cref{sec:overview} provides background in MBSE and describes the energy system design as a motivating example. \Cref{sec:fomalize} formalizes experiment management and \Cref{sec:reasoning} reasoning over experiments.
\Cref{sec:architecture} gives the software architecture, \Cref{sec:evaluation} describes our evaluation, \Cref{sec:background} discusses related work and \Cref{sec:conclusion} concludes with a discussion.

\section{Background \& Use Case}\label{sec:overview}



\paragraph{(Model-Based) Systems Engineering}
SE is a process-driven discipline that focuses on processes managing (and systematically reducing) the risk of (the system to-be-delivered) being (potentially) \emph{wrong} while keeping development efficient~\cite{walden2023systems,haberfellner2019systems}.
Development happens in stages, such as seen in the ISO/IEEE/IEC 15288 reference life cycle standard\footnote{\url{https://www.iso.org/standard/63711.htm}} and the commonly used V-model.
Experiences in SE show that the cost of addressing issues in a system design increases exponentially as the development progresses~\cite{cederbladh2023early, TEKINERDOGAN202145}. Therefore, it is of interest to detect and address any issues as early as possible. Similarly, making design decisions as early as possible is valuable as it reduces overall development time/effort~\cite{madni2018model,gregory2020long}. In many SE contexts this is made increasingly important due to the often high degree of variability which results in large ranges of solutions~\cite{puntigam2020integrated,bilic2019integrated}.
MBSE is the formalization of SE through the use of models as the primary artifacts in SE processes \cite{wymore2018model}. There are many reported incentives for adopting MBSE~\cite{henderson2021value,gregory2020long}, but the one of interest here is that it enables more robust early and continuous V\&V.

\paragraph{Early stage decision-making} The \emph{early stage} of development is difficult to define independently of the concrete SE workflow~\cite{cederbladh2023early,cederbladh2024road}, as it needs to be \emph{early} only related to the rest of the workflow. 
Nonetheless, a few commonalities are shared regardless of specific context: (1) The threshold for acceptable uncertainty is higher than in later development, (2) the goal is pruning the design space but not necessarily making a singular design choice, (3) the prioritization concerns models and knowledge reuse over the contextual model accuracy, and (4) the effort of new model development is low.

Thus, early decision-making or V\&V should focus on reducing consequent effort by re-using models and knowledge from previous endeavors, instead of spending much time in creating new artifacts~\cite{cederbladh2024road}. Additionally, as the early stage inherently entails high uncertainty due to the nature of development, the granularity of analysis is often limited or requires extensive encoding of domain-specific knowledge~\cite{aroonvatanaporn2010reducing}.

Employing MBSE enables the concrete \& explicit formalization of domain knowledge held by practitioners. As a result early decision-making can be done by leveraging models as opposed to engineering know-how and best practices~\cite{madni2018model,gregory2020long}. Likewise, creating and managing models as stand-ins for physical prototypes \& testing can significantly reduce resources invested in the gathering of evidence during decision-making~\cite{puntigam2020integrated,hallqvist2023realization}. 

\paragraph{Use case}
We use the industrial case study of the design and development of an energy system for earth moving machinery in an MBSE context from the AIDOaRt research project~\cite{Cederbladh6985}, which is illustrative of both the need for reuse, and the challenges to implement it. The domain is historically hardware-intensive with rich foundations in product-line engineering for hardware systems~\cite{suryadevara2018adopting,sjoberg2017industrial}, which makes physical experiments necessary but expensive. In this context, system development relies on hardware modularity through interface management~\cite{bruun2015plm}, resulting in highly customizable product configurations~\cite{bilic2019integrated}. As a result, most earth moving machines are \textit{unique} yet belong to larger product families. 

To store and supply energy for these earth-moving machines, large-scale batteries are often integrated with the machine to achieve continuous operation for longer durations without recharging. Large-scale batteries are highly complex and offer wide configuration options during design~\cite{hannan2021battery}. Therefore, battery systems consider many decision-points, starting from estimation of total required energy capacity, followed by cell selection and sizing, total energy system architecture, detailed safety and thermal analysis, lifetime estimations, Thermal Management System (TMS) design, prototyping, etc. The \emph{early} decision-points consider numerous trade-offs and stakeholders, commonly perceived as a bottleneck. For the case study, we extract two sub-stages of energy system design: battery design and TMS design for earth moving machinery, similar to~\cite{Cederbladh6985}.



The architecture of a battery system is modular, to allow standardized battery modules to be connected in series and parallel with little integration efforts. By architecting different topologies, the total capacity, nominal voltage, maximum discharge, and other characteristics are easily parametrized. The main trade-off with a larger battery is the cost, both directly through the monetary value of cells themselves, but also indirectly from increased thermal management needs, mass, and volume. Finding a minimal battery conforming to the requirements of the (earth-moving) machine is both essential and a challenge, needing to satisfy functional demands from the customer use case and safety concerns governed by standards, while having an expected lifetime of 10-20 years. To evaluate battery designs, trade-off simulations are often run with different configurations using parameterized models with standard drive-cycles as seen in \Cref{fig:battery}, the parameters being \emph{BatteryVoltage}, \emph{MaxTorque}, etc. These simulations of different battery configurations are an important step in determining the sizing and scope of the battery, especially in the early system design phases. A battery configuration is considered unstable if its state of charge $\mathsf{SoC}$ drops below 50\% at any point during the drive cycle. \Cref{fig:battery-stable-unstable} shows the evolution of the state of charge of the battery system with stable and unstable configurations.

\begin{figure}[t]
\centering
\includegraphics[width=\linewidth]{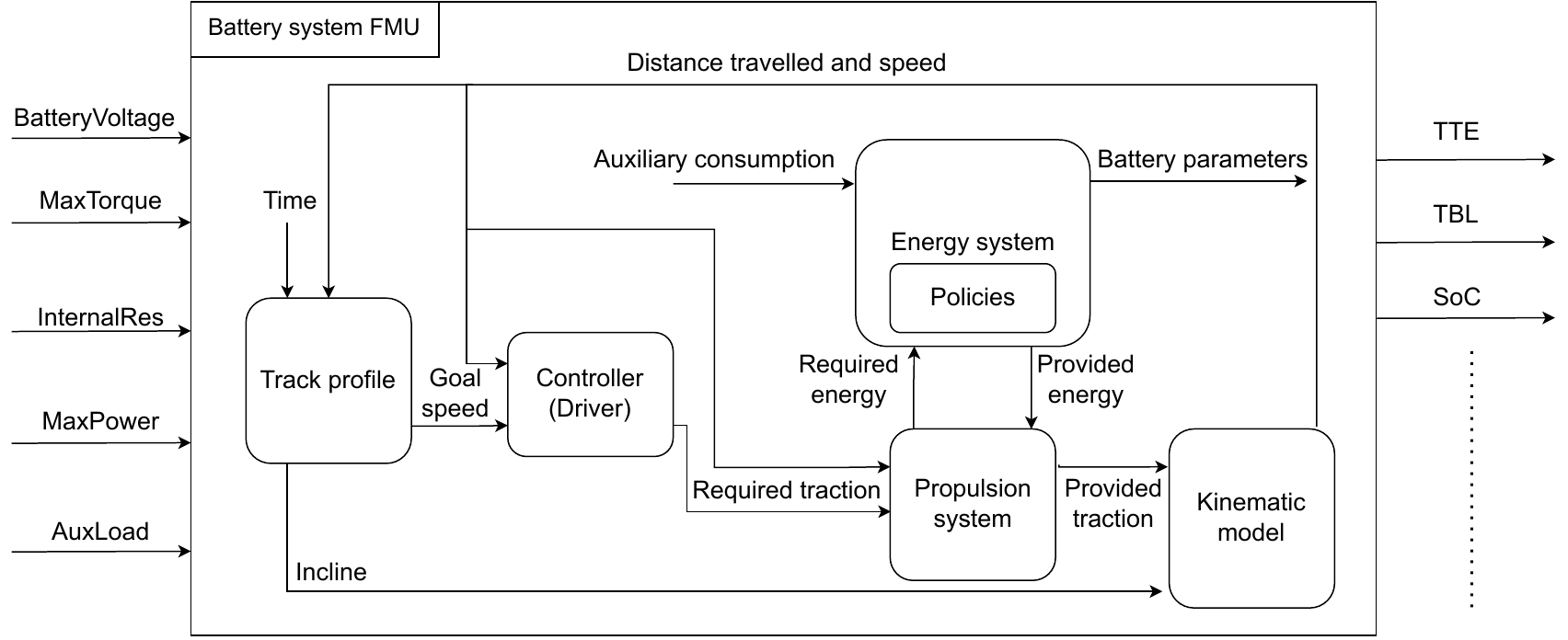}
\caption{High-level architecture architecture of a battery FMU used in early analyses, from the case-study.}
\label{fig:battery}
\end{figure}



\begin{figure}
	\centering
	\includegraphics[width=\linewidth]{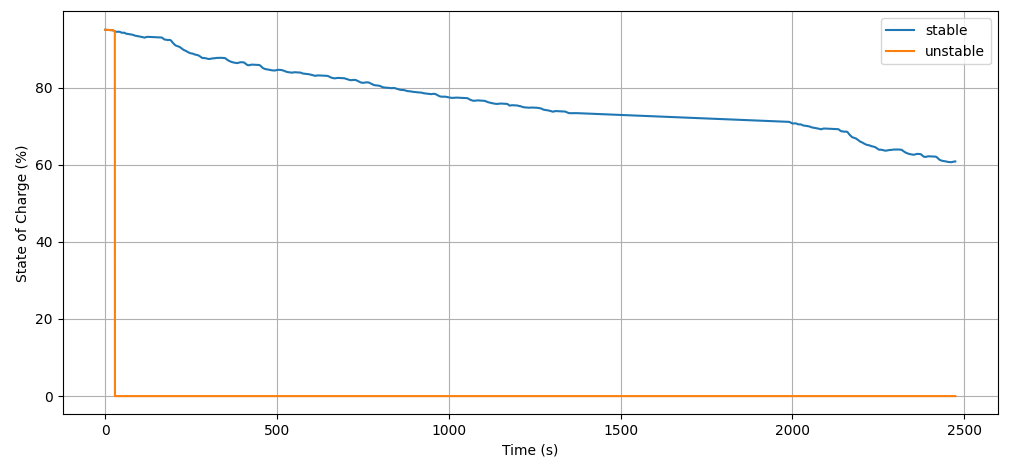}
	\caption{Evolution of $\mathsf{SoC}$ during simulation of the battery FMU in a standard drive cycle with (a) stable parameters ($\mathsf{Voltage} = 234.4V, \mathsf{MaxTorque} = 800.42Nm, \mathsf{InternalRes} = 0.03 \Omega$) and (b) unstable parameters ($\mathsf{Voltage} = 340.95V, \mathsf{MaxTorque} = 998.46Nm, \mathsf{InternalRes} = 0.50 \Omega$).}
	\label{fig:battery-stable-unstable}
\end{figure}

\paragraph{Challenge}
A common case encountered is that an engineer wants to reuse experimental results or experimental models from previous evaluations in early stages to reason about viability of energy system solutions. However, in our experience there is rarely a previous experiment that has the exact same conditions or system-under-study; instead there might be similar experiments or setups for some related variant configurations. This also causes issues in cross-domain concerns, where engineers have heterogeneous skill-sets (e.g. skills required for the battery system configuration are not the same as for TMS), leading to parallel work with many manual interaction points. To resolve this, the engineer needs to identify models or experiments and understand whether they are \textit{similar enough} to be reused for a given evaluation. Several aspects come into play: How can this distance between experiments be measured explicitly? How do the factors of uncertainty at the early stage affect the decision-making process? For the engineer, these aspects might often be embedded implicitly in domain expertise and already applied through various processes. 

For the energy system design, the following needs are formulated from the engineering perspective for our proposed framework: (1) The framework should capture implicit knowledge by experts through formalization to the degree that experiments can be compared explicitly, and (2) the framework should provide automated reasoning mechanisms to support reuse of experiments, and more importantly judge when they cannot be reused.
These goals act as the requirements to drive the workflow to an acceptable status given the engineering needs from the industrial context. 
Further requirements are that it should guide software developers in implementing the solution and be integrated in the workflow.

\section{Formalizing Experiment Management}\label{sec:fomalize}
Reasoning about experiments, their data and specifications for reuse and management requires a precise formulation of (a) the overall workflow -- from the user asking a query to execution of experiments, (b) the supposed relations between the components and concepts in this workflow, and (c) the properties of these components and concepts. 
The overall workflow is structured in two semantically opposing directions:

\begin{figure}[bt]
    \centering
    \includegraphics[width=0.75\linewidth]{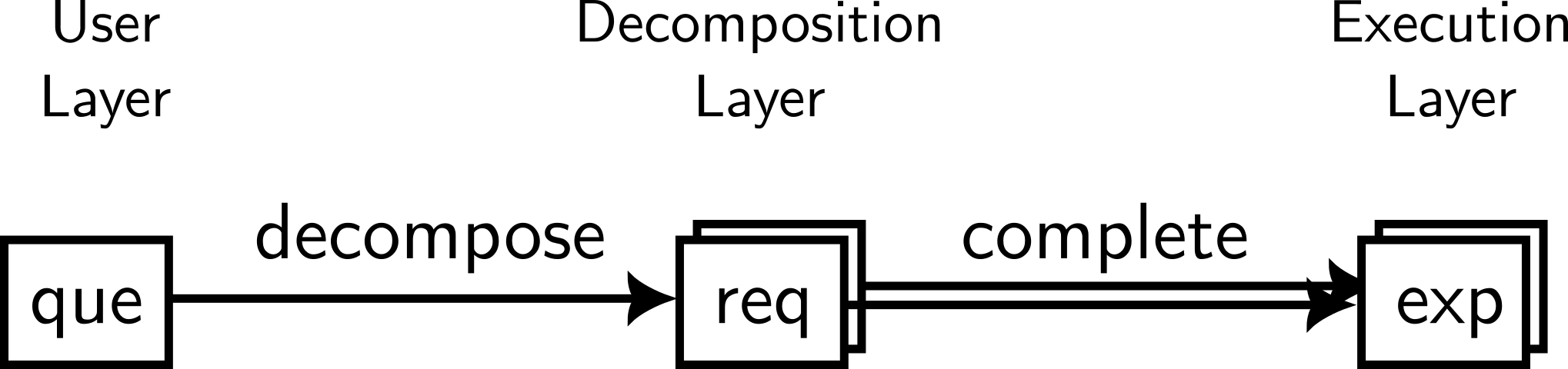}
    \caption{Decomposition and reduction of uncertainty.}
    \label{fig:dir1}
\end{figure}

\begin{itemize}
\item 
From the user query to the experiment specification and execution/deployment -- this direction centers on \emph{reducing uncertainty} in the deployment specification of experiments. The user layer reduces the uncertainty in the user query to generate relevant experiment requests; secondly, the decomposition layer further interprets experiment requests into concrete, deployable experiment specifications, thereby further reducing uncertainty. Thus, in this direction, the query is ultimately concretized and decomposed into a set of experiment specifications. This is illustrated in \Cref{fig:dir1}.
\item
The reverse direction, from experiment result to final answer -- this direction centers on \emph{reuse} of results and responses related to experiments. If a query is decomposed into experiment requests that match or are similar to (or inferable from, via CBR) requests posed in the past, then the specific layer reuses past information to compute the answer instead of delegating to the next lower-level layer. This is true for both This is illustrated in \Cref{fig:dir2}.
\end{itemize}

\begin{figure}[tb]
    \centering
    \includegraphics[width=0.75\linewidth]{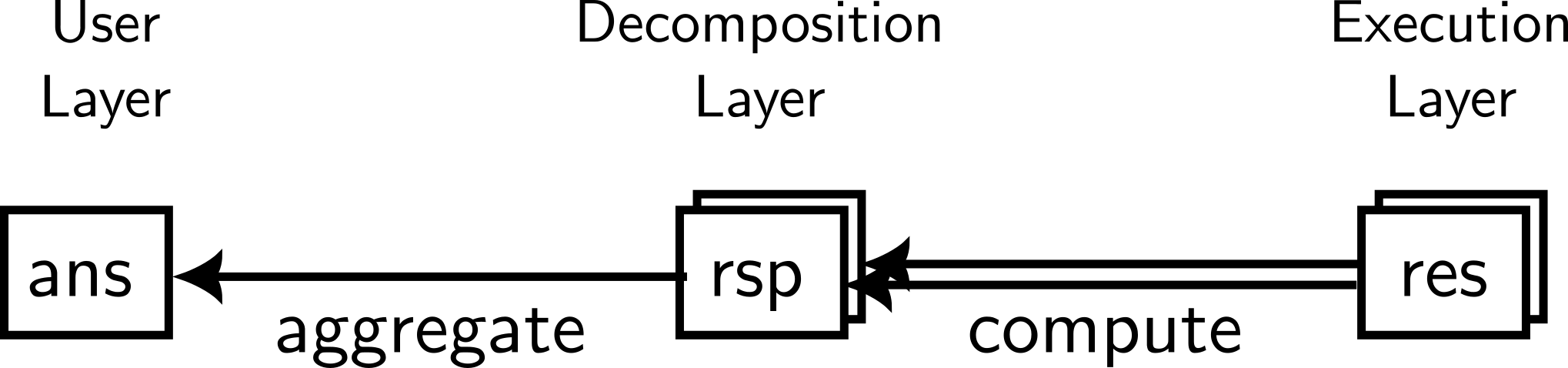}
    \caption{Reuse and composition.}
    \label{fig:dir2}
\end{figure}

We turn to queries of reuse, storage and data management in \Cref{sec:reasoning,sec:architecture} and first establish the properties and relations around  queries and experiments. We leave some elements, e.g., what exactly constitutes a ``validity frame''~\cite{AckerMVD24, mittal23} of an experiment, to be instantiated for concrete applications. However, we specify conditions that must be fulfilled to enable reuse. 

\subsection{Running example}
\label{sec:example}

To succinctly explain the tripartite experiment structures in our framework, we will use a highly simplified example of simulations of the physics to measure the braking distance of a train on a sloped track with friction. The Free-Body Diagram (FBD) of the train in \Cref{fig:fbd-train} depicts the forces acting on the train. From the FBD, the constant acceleration $a$ (negative deceleration) of the train is calculated as:

$$ a = - ( m \cdot g \cdot sin(\theta) + F_B + \mu \cdot N) $$

From the second equation of motion, $v^2 = u^2 + 2\cdot a \cdot s$, where $v (= 0)$ = final velocity, $u$ = initial velocity, the displacement $s$ can be computed as (by substituting the value of $a$ from above):

$$ s = \frac{u ^ 2}{2 \cdot (F_B + m \cdot g \cdot (sin(\theta) + \mu \cdot cos(\theta)))}$$

The displacement $s$ is directly proportional to the initial velocity $u$, and inversely proportional to the mass of the train $m$, braking force $F_B$, slope $\theta$, and coefficient of friction $\mu$ (when $0 \leq \mu \leq 1$ and $-10 \degree \leq \theta \leq 10 \degree$)\footnote{\url{https://en.wikipedia.org/wiki/List_of_steepest_gradients_on_adhesion_railways}}. A trace over time of the velocity and displacement from one simulation of the train example is presented in \Cref{fig:train-trace}.

\begin{figure}[tb]
    \centering
    \includegraphics[width=\linewidth]{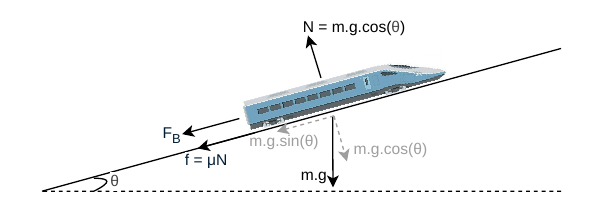}
    \caption{Free-body diagram of simplified train running example with $N$ = Normal force, $m$ = mass of train, $g$ = acceleration due to gravity on Earth's surface, $\theta$ = average slope of track, $F_B$ = braking force, $f$ = friction force, $0 < \mu < 1$ = friction coefficient.}
    \label{fig:fbd-train}
\end{figure}

\begin{figure}
     \centering
     \begin{subfigure}[b]{0.49\textwidth}
         \centering
         \includegraphics[width=\textwidth]{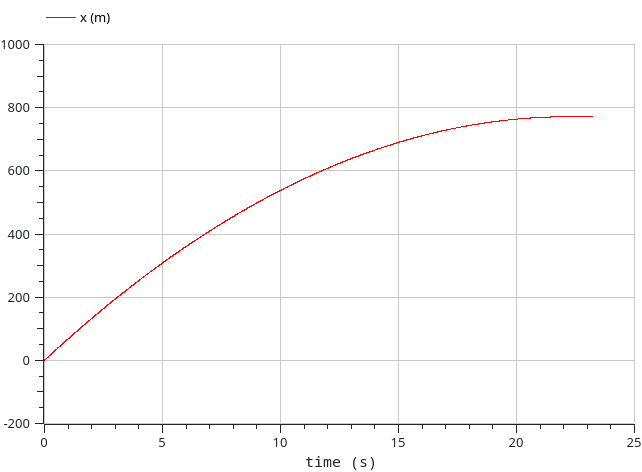}
         \caption{Displacement of train over time}
     \end{subfigure}
     \begin{subfigure}[b]{0.49\textwidth}
         \centering
         \includegraphics[width=\textwidth]{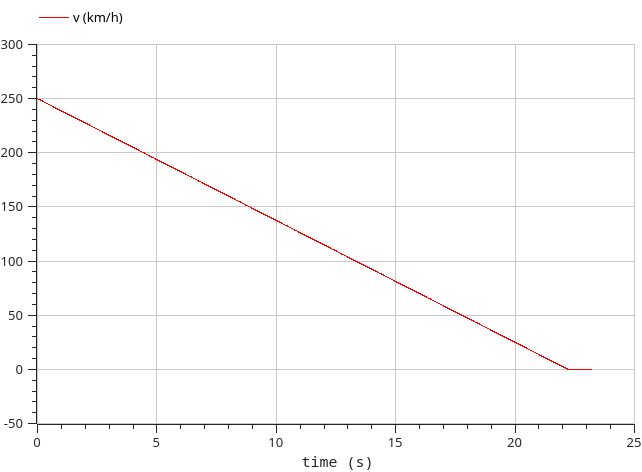}
         \caption{Velocity of train over time}
     \end{subfigure}
        \caption{Sample traces from one simulation of the train braking system.}
        \label{fig:train-trace}
\end{figure}

\subsection{Tripartite Experiment Management Structure}



A user query can be asked from different perspectives, e.g., from the perspective of an engineer designing the train, or from the perspective of a sales person interested in high level queries like safety. 
While we can relate a query to multiple corresponding experiments (exact match, fuzzy match, or CBR-inferred match), the relation of query and experiment does not substantiate a full \emph{specification} of the experiment, but only an abstract \emph{request}.
We derive three requirements from this: First, we must distinguish between different users and their queries; second, we must distinguish between a user query and its decomposition into experiment requests; and third, we must distinguish between an experiment request and a full experiment specification.

Consequently, our framework is structured in three layers.
Each layer is defined by a \emph{language}, which defines, on an abstract level, (a) possible calls to that layer as described in the \emph{decomposition} direction in \Cref{fig:dir1} and (b) replies to the calls as described in the \emph{reuse} direction in \Cref{fig:dir2}. This is illustrated in \Cref{fig:store}.

\begin{description}
    \item[User Layer] Users have access to a set of parametric \emph{queries} \quest that defines their perspective or view on the System-under-Study (SuS), and for each kind of query, expect a certain kind of \emph{answer} \answer. The query language \qlang of these queries is geared towards the user and not the low-level experiment execution.
    \item[Decomposition Layer] Internally, a user query may correspond to one or more experiment \emph{requests} \request. In contrast to a query, a request is a technical view on the experiments, not the SuS. Indeed, a query may be decomposed into several requests, e.g. a user query . The language of requests \erlang is geared towards describing an external view on experiments. Each request is fulfilled by exactly one \emph{response} \response.
    \item[Execution Layer] Finally, an experiment must be fully specified in order to be excuted. In contrast to a request, an \emph{experiment specification} \experiment fully describes the experiment, including internal structure such as (physical) setup.
    Each experiment specification results in one experiment \emph{result} \exresult. An experiment specification is an extended request defined by an experiment scheme language \exscheme.
\end{description}

The framework can be instantiated with multiple query languages for different properties and user groups, while the composition and execution layers are uniform throughout the application. In the following, we investigate each layer and establish the connections between them. 
    
\begin{figure}[b]
    \centering
    \includegraphics[width=0.65\linewidth]{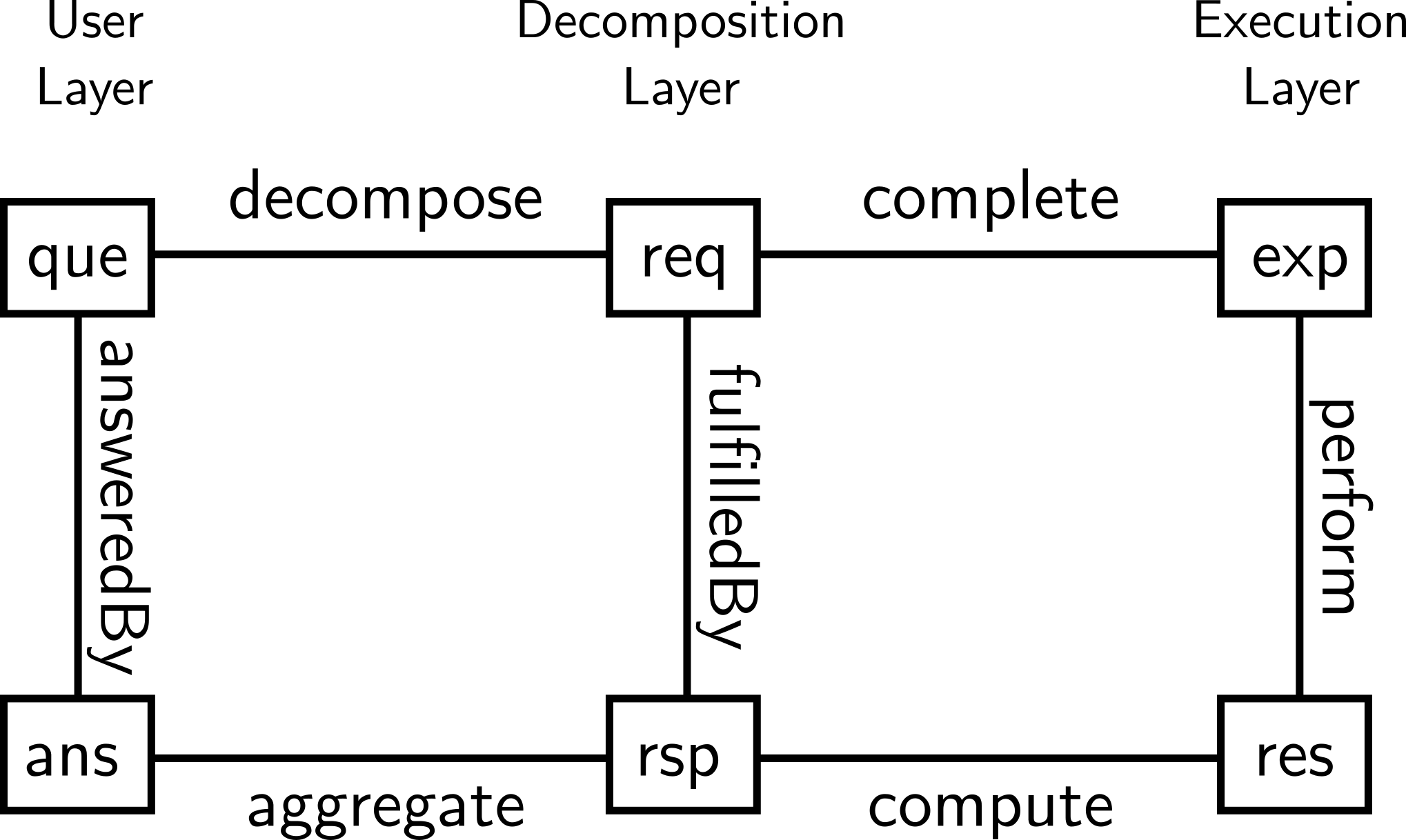}
    \caption{Complete structure.}
    \label{fig:store}
\end{figure}

\subsection{User Layer}
A query language \qlang defines the variables that the user can instantiate, as well as those combinations of variables that define a valid query. Such combinations are called query schemes.
As discussed above, we assume that each query language is specific to one perspective, property and user group.

\begin{definition}[Query Languages]
A query language \qlang is a tuple 
\[
\langle V_\qlang, \mathcal{L}_\qlang, \Delta_\qlang, A_\qlang\rangle
\]
where $V_\qlang \neq \emptyset$ is the set of query variables describing the subject and context of an experiment.
Furthermore, $\mathcal{L} \subseteq 2^{V_\qlang}$ is the set of query schemes, $\Delta_\qlang$ the set of maps from query variables to sets of possible values, essentially providing a typing, and $A_\qlang$ the set of possible values for the answer returned by the user layer.
\end{definition}

The set of query variables is the set of all possible parameters the user can choose. However, it is not necessary for the user to provide all of them, either because not all combinations are meaningful or because some parameters may be deliberately under-specified possibly matching with multiple experiment requests in the decomposition layer. Each meaningful combination of variables is a query scheme. In particular, we point out that some variables determine the properties of interest.


\begin{example}\label{ex:trainlang1}
We model a query language for our running example of train braking, specifically from an engineers view, and denote it by the subscript $_{\mathsf{eng}}$. The query posed by the engineer via this language is, "Does the train with a certain mass and braking power, in a certain environmental condition of initial velocity, track friction, and slope, stop within a certain distance when braking force is applied uniformly?"
\[
V_{\mathsf{Train}_{\mathsf{eng}}} = \{m, F_B, v, \mu, \theta, \mathit{dist}\} \qquad \qlang^\mathsf{Train}_{\mathsf{eng}} = 
\langle V_{\mathsf{Train}_{\mathsf{eng}}}, \{V_{\mathsf{Train}_{\mathsf{eng}}}\}, V_{\mathsf{Train}_{\mathsf{eng}}}\!\rightarrow\mathbb{R}, \mathbb{B}\rangle
\]
\end{example}
Here, the engineer may define weight ($\mathit{m}$), braking force ($F_B$), start velocity ($\mathit{v}$), track friction ($\mu$), the slope ($\theta$), and the expected stopping distance ($\mathit{dist}$). Indeed, the query scheme $\{V_{\mathsf{Train}_{\mathsf{eng}}}\}$ says that all variables are required to be defined. Each variable must be assigned a Real number value and the answer will be a Boolean value namely whether the trains stops within the specified distance or not.

\begin{example}\label{ex:trainlang2}
The second example takes the higher-level view of a salesperson, and denote it by the subscript $_{\mathsf{sale}}$. They question whether a certain train (without considering the specification of a train configuration like weight or braking force, but simply a finished train product) is safe according to some regulations (that usually pre-define an environmental situation with a maximum stopping distance).  


\begin{align*}
&\qlang^\mathsf{Train}_{\mathsf{sale}} = \big\langle \{\mathit{train}, \mathit{situation}\},  \{\{\mathit{train}, \mathit{situation}\}, \{\mathit{train}\}\}, \Delta^\mathsf{Train}_{\mathsf{sale}}, \mathbb{B}\big\rangle\\
&\Delta^\mathsf{Train}_{\mathsf{sale}}(\mathit{train}) = \{t_1,t_2,t_3\} \qquad 
\Delta^\mathsf{Train}_{\mathsf{sale}}(\mathit{situation}) = \{\mathsf{lvl}_1,\mathsf{lvl}_2\}
\end{align*}
\end{example}
In this example, there are more abstract parameters than the engineer's query language. The $\mathit{train}$ parameter refers to some kind of train version ($t_1$, $t_2$, or $t_3$) and the $\mathit{situation}$ parameter refers to some predefined situation ($\mathsf{lvl}_1$, $\mathsf{lvl}_2$). The user may choose to not specify the situation, in which case the intended meaning (only specific to this example) is that the chosen train is safe in all predefined situations ($\mathsf{lvl}_1$ and $\mathsf{lvl}_2$).

\begin{definition}[Queries]
Formally, a query is a function from query variables to values, such that the domain of this function is one of the query schemes. Let \qlang be a query language. The set of queries for this language is defined as follows.
\[
\mathcal{Q}_\qlang = \{\quest ~|~  \exists l \in \mathcal{L}_\qlang.~\mathbf{dom}~\quest = l, \forall v \in \mathbf{dom}~\quest.~\quest(v) \in \Delta_\qlang(v) \}
\]
\end{definition}

An answer is merely a value\footnote{It is a value in the sense that it is a member of $\Delta_a$, it can be a complex structure on its own, e.g., a distribution. In such a case, $\Delta_a$ would be the set of possible distributions.}, fitting the set $\Delta_a$.
Queries are related to the answers through the predicate \(\answered(\quest,\answer) \subseteq \mathcal{Q}_\qlang \times A_\qlang\)

\begin{example}
The following are queries for the query languages in \Cref{ex:trainlang1,ex:trainlang2}, with
$\quest^\mathsf{train}_1, \quest^\mathsf{train}_2 \in \qlang^\mathsf{Train}_{\mathsf{eng}}$ and
$\quest^\mathsf{train}_3, \quest^\mathsf{train}_4 \in \qlang^\mathsf{Train}_{\mathsf{sale}} $
\begin{itemize}
    \item \textit{``Does the train with weight 1500$T$, brake force of 0.12$kN/T$, and initial velocity of 120$km/h$, stop within 1600$m$ on a track with slipperiness $0.7$? and slope of +10 $\degree$''}
    \[
    \quest^\mathsf{train}_1 = \{m \mapsto 1500, F_B\mapsto 0.12, v\mapsto 120, \mathit{slip}\mapsto 0.7, \theta\mapsto 10, \mathit{dist}\mapsto 1600\} 
    \]
    \item \textit{``Does the train with weight 20000$T$, brake force of 0.09$kN/T$, and initial velocity of 200$km/h$, stop within 1000$m$ on a track with slipperiness $0.1$? and slope of -10 $\degree$''}
    \[
    \quest^\mathsf{train}_1 = \{m \mapsto 20000, F_B\mapsto 0.09, v\mapsto 200, \mathit{slip}\mapsto 0.1, \theta\mapsto -10, \mathit{dist}\mapsto 1000\} 
    \]
    \item \textit{``Does train $t_1$ fulfill the safety requirements?''}
    \[
    \quest^\mathsf{train}_3 = \{\mathit{train}\mapsto t_1 \}
    \]
    \item \textit{``Does train $t_1$ fulfill the safety requirements described by scenario $\mathtt{lvl}_1$?''}
    \[
    \quest^\mathsf{train}_4 = \{\mathit{train}\mapsto t_1, \mathit{situation}\mapsto\mathtt{lvl}_1 \}
    \]
\end{itemize}
\end{example}

\subsection{Decomposition Layer}
While the user layer is concerned with the parameters of the queries asked by a user, the \emph{decomposition layer} is concerned with the parameters of the experiments associated with a query and the mapping from query to request. In the extreme case, namely, a query language that corresponds to a very technical, internal view of the experimental set up - these two views may overlap; in general however, we can see the user layer as the interface to the higher-level user view, and the decomposition layer as a lower-level view on experiments.

Formally, the decomposition layer is based on an \emph{experiment request language}, \erlang, which defines experiment requests and their responses.
\begin{definition}[Experiment Request Language]
    An experiment request language \erlang is a tuple
    \[
    \langle V_\erlang, I_\erlang, \Delta_\erlang\rangle
    \]
    Where $V_\erlang$ are the request variables for the experiment request and 
    $I_\erlang$ is the set of property of interest variables. 
    $\Delta_\erlang$ is a map from request variables and property of interest variables to sets of their possible values.
\end{definition}

\begin{example}\label{ex:trainlang3}
The experiment request language for our train running example is the following.
\begin{align*}
\erlang^\mathsf{Train} &= \big\langle \{(m, F_B, v, \mu, \theta, \mathit{dist})\}, \{\mathit{stopDist}\}, \Delta_{\erlang^\mathsf{Train}}  \big\rangle\\
\Delta_{\erlang^\mathsf{Train}}(x) &= \mathbb{R}\text{ $\forall x \in V_\erlang \cup I_\erlang$}
\end{align*}    
\end{example}

\begin{definition}[Requests and Responses]
Let \erlang be a request language.
    The set of experiment requests is defined as
    \[
    \mathcal{R}_\erlang=
    \{
        \request ~|~ \mathbf{dom}~\request = V_\erlang, \forall v \in \mathbf{dom}~\request.~\request(v) \in \Delta_\erlang(v)
    \}
    \]
    The set of experiment responses is defined as
    \[
    \mathcal{D}_\erlang=
    \{
        \response ~|~ \mathbf{dom}~\response = I_\erlang, \forall v \in \mathbf{dom}~\response.~\response(v) \in \Delta_\erlang(v)
    \}
    \]
    We say that a request is fulfilled by a response and denote this relation with $\fulfilled(\request,\response)$. Note that the response corresponding to an experiment request is a value -- to answer the query of the user, all such responses that may be associated with the user query must be aggregated.
\end{definition}

An experiment request does not fully describe an experiment -- it is a \emph{partial experiment specification} that specifies enough of the experiment to describe its relation to the user query. The main task of the decomposition layer is to describe experiments sufficiently precise to decompose queries and compose answers. To formalize these requirements, we employ two functions.

\begin{definition}[Decomposition]
A query decomposes into a set of requests.
\[
\decomp: \mathcal{Q}_\qlang \rightarrow 2^{\mathcal{R}_\erlang}
\]
A set of responses is aggregated into a single answer.
\[
\aggregate: 2^{\mathcal{D}_\erlang}\times\mathcal{Q}_\qlang \rightarrow A_\qlang
\]
\end{definition}

Aggregation is with respect to the original query -- as the same request language can be used for different query languages, this information is needed to, e.g., select the correct property of interest variable. The general condition we pose on fulfillment of decomposition and aggregation is the following.
\begin{definition}[Compatible User and Decomposition Layers]\label{def:compat}
Let \qlang be a query language and \erlang be an experiment request language.
Let \answered, \fulfilled, \decomp, and \aggregate be some answering, fulfillment, decomposition and aggregation function, respectively.
We say that \qlang, \erlang and \answered, \aggregate, \decomp, \fulfilled are \emph{compatible} if the following holds for all $\quest \in \mathcal{Q}_\qlang$.
    \[ \answered\Big(\quest,~\aggregate\big(F, \quest\big)\Big)\]
where $F = \{ \response ~|~\exists \request \in \decomp(\quest).~ \fulfilled(\request,\response)\}$.
\end{definition}
This condition expresses that decomposing a query into a set of requests, generating fulfilling responses for all of them (set $F$), and then aggregating these responses' results in an answer to the original query.

\begin{example}\label{ex:trainlang4}
    Let us continue \Cref{ex:trainlang3} by examining $\qlang^\mathsf{Train}_\mathsf{eng}$, the query language for the engineer's view. We define the decomposition and aggregation relations as follows (The subscript $_0$ signifies the value specified by the user, in the range of $\Delta^{\mathsf{train}}_{eng}$):
\begin{align*}  
\decomp^\mathsf{Train}_\mathsf{eng}&(\{m \mapsto m_0, F_B\mapsto F_{B0}, v\mapsto v_0, \mu\mapsto \mu_0, \theta \mapsto \theta_0, \mathit{dist}\mapsto \mathit{dist}_0 \}) =\\
&\{
\{m \mapsto m_0, F_B\mapsto F_{B0}, v\mapsto v_0, \mu\mapsto \mu_0, \theta \mapsto \theta_0\}
\}
\\
\aggregate^\mathsf{Train}_\mathsf{eng}&(\{\mathit{stopDist} \mapsto s\}) \iff s < \mathit{dist}_0
\end{align*}
\end{example}

In the above decomposition, every query from the engineer is mapped directly to an experiment request; note that we do not include the query variable $dist$ (the expected maximum stopping distance), since it is needed only to compute the response, and not to describe the experiment. 

Consequently, aggregation merely transforms the result value by comparing it with the maximum stopping distance specified by the user, to a Boolean value.

\begin{example}\label{ex:trainlang5}
    Let us now examine $\qlang^\mathsf{Train}_\mathsf{sale}$, the query language for the salesperson's view, in the context of \Cref{ex:trainlang3}, which has a more complex decomposition and aggregation. The subscript $_0$ signifies the value specified by the user, in the range of $\Delta^{\mathsf{train}}_{sale}$.
\begin{align*}  
\decomp^\mathsf{Train}_\mathsf{sale}&(\{\mathit{train} \mapsto t_0 \}) =\\
\{
&\{m \mapsto m(t_0), F_B\mapsto F_B(t_0), v\mapsto v(\mathsf{lvl}_1), \mu\mapsto \mu(\mathsf{lvl}_1), \theta \mapsto \theta(\mathsf{lvl}_1)\},\\
&\{m \mapsto m(t_0), F_B\mapsto F_B(t_0), v\mapsto v(\mathsf{lvl}_2), \mu\mapsto \mu(\mathsf{lvl}_2), \theta \mapsto \theta(\mathsf{lvl}_2)\}
\}
\end{align*}
\begin{align*}
\decomp^\mathsf{Train}_\mathsf{sale}&(\{\mathit{train} \mapsto t_0, \mathit{situation} \mapsto \mathsf{lvl}_0 \}) =\\
\{
&\{m \mapsto m(t_0), F_B\mapsto F_B(t_0), v\mapsto v(\mathsf{lvl}_0), \mu\mapsto \mu(\mathsf{lvl}_0), \theta \mapsto \theta(\mathsf{lvl}_0)\}
\}
\end{align*}
\begin{align*}  
\aggregate^\mathsf{Train}_\mathsf{sale}((\{\mathit{stopDist} \mapsto s_i\})_{i\in I}) \iff \forall i \in I.~s_i < min(dist(\mathsf{lvl}_1), dist(\mathsf{lvl}_2))
\end{align*}

\end{example}
There are two decomposition functions associated with the salesperson's view. In the first one, where the \emph{situation} is under-specified, a query corresponds to multiple experiment requests, in this case two, one for each situation $\mathsf{lvl}_1$ and $\mathsf{lvl}_2$. On the other hand, if the query is fully specified, it corresponds to exactly one request. 

Aggregation now encodes the domain knowledge that returns a Boolean describing whether the train stops within the minimum value of the stopping distance specified by all pre-defined \emph{scenarios}, in case there are multiple corresponding experiment requests (that is, if the \emph{scenario} is under-specified in the user query).

The above examples define the decomposition and aggregation functions. The former operates on queries and requests, while the later operate on answers and responses. To link these two different levels, one must also define the $\answered^\mathsf{Train}_\mathsf{sale}, \fulfilled^\mathsf{Train}_\mathsf{sale}$ relations (for \cref{ex:trainlang5}). These relations track which requests and responses, as well as which queries and answers belong together, to keep different queries and language stack separated. Their exact definition is determined by the storage and the compatability condition of ~\cref{def:compat}, not by reasoning or reuse, and we thus refrain from giving it here.

\subsection{Execution Layer}
The execution layer is concerned with adding enough information to an experiment request to so that it can be performed and is replicable.
As such, its main elements are full experiment specifications. 
We treat experiment specifications and the results of executed experiments abstractly, as their inner workings and structure are not within the scope of our framework since we are agnostic to the precise schema used for experiment specifications.
However, we do observe that experiments rarely come alone and require that given a set of experiments, there is a partial order between them that expresses the order in which they must be performed. This abstracts from details on planning a series of experiments.

\begin{definition}
    An experiment specification language $\exscheme$ is a tuple
    \[
    \langle\experiments, \leq_\exscheme, \mathsf{Res}_\exscheme
    \rangle 
    \]
    where $\experiments$ is the set of possible experiment specifications, $\leq_\exscheme$ is a partial order on $\experiments$, and $\mathsf{Res}_\exscheme$ is the set of experiment results.
\end{definition}

\begin{definition}
Let \exscheme be an experiment specification language. We say that \emph{executing} an experiment with specification $\experiment \in \experiments_\exscheme$ results in $\exresult \in \mathsf{Res}_\exscheme$. This is denoted as $\execute(\experiment,\exresult)$.
\end{definition}

\begin{example}\label{ex:trainlang6}
Continuing \Cref{ex:trainlang4,ex:trainlang5}, the results for the train experiments are functions $\mathbb{R}^+ \rightarrow \mathbb{R}$ that represent the trace of displacement of the train as a function of time. A sample trace is presented in \Cref{fig:train-trace}(a).
\end{example}

Given a set of experiment requests, the experiments are not just mapped to experiment specifications, but to a `partially ordered set of experiment specifications', where the order expresses the plan of executing the experiments. This can be a low-level form of reuse (e.g. reuse of resources for subsequent experiments), but not on the level of data or reasoning, as is the focus of this paper. To connect the execution layer to the decomposition layer, we formalize the `completion' into an experiment specification, of an experiment request, as follows. 

\begin{definition}[Completion]
A request is `completed' into a set of experiment specifications.
Let $\request \in \mathcal{R}_\erlang$ be an experiment request. We say that $\experiment \in \experiments_\exscheme$ completes \request, written $\complete(\request,\experiment)$, if \request uniquely determines \experiment.
\[
\complete: \mathcal{R}_\erlang \rightarrow 2^{E_\exscheme}
\]
    
Let \erlang be an experiment request language and $\mathsf{Res}$ a set of experiment results.
A computation function is a function from results to responses.
\[
\compute: 2^{\mathsf{Res}_\exscheme} \times \mathcal{R}_\erlang \rightarrow \mathcal{D}_\erlang 
\]
\end{definition}

Experiments are assumed to be deterministic.
As $\leq_\exscheme$ is defined over all experiment specifications, every subset of experiment specifications is also partially ordered:
Applying $\complete$ to a set of requests suffices to this end. It is also important to note that the $\compute$ function depends on the property-of-interest variable ($\in I_\erlang$) specified in the \request. The \compute function can be invoked over multiple experiment results for example, when aggregating results from multiple trials of the same experiment.

\begin{example}\label{ex:trainlang7}
In the train running example, an experiment response (the stopping distance of the train) can be computed from the time series trace (experiment result) with 
\[
\compute_\mathsf{Train}(\exresult) = \{ \mathit{stopDist} \mapsto max(\exresult(x)) \} 
\]
\end{example}

The \compute function above simply assigns the maximum value of the displacement trace of the train, $x$, to the stopping distance $stopDist$ property of interest.

We note that if an application based on our framework needs to access specific properties associated with experiment specifications, this can be easily done by defining abstract functions over $\experiments$. For example, validity frames (defined as a set of experiment specifications where a model is deemed a valid representation of the system-under-study w.r.t specific properties of interest \cite{mittal23}) can be expressed as a function that retrieves a validity frame, and a boolean predicate whether a validity frame subsumes another. Let $\mathsf{Valid}$ be the set of validity frames w.r.t.\ some experiment specification language \exscheme.
\begin{align*}
    \mathsf{validity\textrm{--}frame}&: \experiments \rightarrow \mathsf{Valid}\\
    \mathsf{validity\textrm{--}subsume} &\subseteq \mathsf{Valid} \times \mathsf{Valid}
\end{align*}

Next, we impose the essential conditions of compatibility of composition and execution layers as we did for user and decomposition layers.
\begin{definition}[Compatible Decomposition and Execution Layers]
Let \erlang be an experiment request language and \exscheme an experiment specification language
Let \compute, \fulfilled, \execute, and \complete be some computing, fulfillment, execution and completion function, respectively.
We say that \erlang, \exscheme and \compute, \fulfilled, \execute, \complete are compatible if the following holds for all $\request \in \mathcal{R}_\erlang$:
    \[
    \fulfilled(\request, \compute(\execute(\complete(\request))))
    \]
\end{definition}

As expected, the compatibility notions can be checked in isolation and combined to the expected property of the overall languages:
Performing all experiments completed from the decomposition of a query results in a set of results that can be used to compute a set of responses that aggregate to an answer to the original query. Note that these languages may be the same -- for example, the requests may be (less uncertain) queries.
\begin{definition}
    A language structure $\mathit{LS}$ is a tuple of a query language, an experiment request language, an experiment specification language, 
    equipped with fitting \compute, \fulfilled, \execute, \complete, \aggregate, and \decomp functions. 
    \[ \mathit{LS} = \langle \qlang,\erlang,\exscheme \rangle \]
    We say that a language structure is compatible if the following holds for every $\quest \in \mathcal{Q}_\qlang$.
    \begin{align*}
    &exp = \left\{ \response ~\mid~ \exists \request \in \decomp(\quest).~\response = \compute(\request, \execute(\complete(\request))) \right\}\\
    &\forall \quest.~\answered\left(\quest, \aggregate\left(\quest, exp\right)\right)
    \end{align*}    
\end{definition}

\begin{lemma}
Let $\mathit{LS} = \langle \qlang,\erlang,\exscheme \rangle$ be a (tripartite) language structure. $\mathit{LS}$ is compatible if (1) $\qlang$ and $\erlang$ are compatible and (2) $\erlang$ and $\exscheme$ are compatible.
\end{lemma}
\begin{proof}
    We have to show that $\forall \quest.~\answered\left(\quest, \aggregate\left(\quest, X\right)\right)$, where 
    \[X = \left\{ \response ~\mid~ \exists \request \in \decomp(\quest).~\response = \compute(\request, \execute(\complete(\request))) \right\}.\]
    By compatibility (1), we have that this holds exactly if the following holds.
    \[X =  \{\response ~|~\exists \request \in \decomp(\quest).~ \fulfilled(\request,\response)\}\]
    Thus, we have to show equality of these sets, i.e., that 
    \[
    \response = \compute(\request, \execute(\complete(\request))) \iff \fulfilled(\request,\response)
    \]
    This is exact compatibility (2). \qedhere
\end{proof}

\noindent We only consider compatible language structures henceforth.
\section{Reasoning over Experiments}\label{sec:reasoning}
Armed with our formalization of experiment structures, we now formulate several kinds of reuse. Our framework utilizes \emph{intelligent} and \emph{robust} caches for all three layers in the tripartite experiment structure.
Caching enables simple retrieval and reuse of a returned value (resp.\ answer, response, or result) if the same input (resp.\ query, request or specification) was given to that specific layer. 
Already here, the layered architecture enables additional reuse: If a query is decomposed into several requests, it may be the case that \emph{some} of the requests are obtained directly from the cache while others are actually executed.

\begin{description}
    \item[Intelligent Caching] With intelligent caching, each layer is equipped with \emph{Case-Based Reasoning} (CBR) capabilities that attempt to derive, for example on the decomposition layer, a response from a previously encountered, but different, request without communicating with the lower-level layer, i.e., without passing an experiment specification to the execution layer in the example.
    \item[Robust Caching] With robust caching, each layer is equipped with two distances, which realize \emph{robust retrieval}. The first distance defines when, for example, for the decomposition layer, a request is considered similar enough to a previous request that the response can be reused directly. The second distance defines when a request is considered similar enough to a previous request that the response cannot be reused directly, but \emph{computed} from the result of the same underlying data.
\end{description}

We introduce reuse through intelligent caching in three steps. \Cref{ssec:store} introduces an abstract notion of an experiment store, \Cref{ssec:reason} introduces the formalisms for CBR and distances, and describes the retrieval steps, and \Cref{ssec:algo} gives the full algorithm and shows its correctness.

\paragraph*{Preliminaries}

Our reuse framework is based on `distances' over experiment-related structures. A distance over a set $X$ is defined as $\dist{\cdot,\cdot}:X\times X \rightarrow \mathbb{R}^+$ such that the following condition holds:
$\forall x.y \in X.~\dist{x,x} = 0\wedge
    \dist{x,y} = \dist{y,x}\wedge
    \dist{x,z} \leq \dist{x,y} + \dist{y,z}
$.   
A partial order $\leq$ over a set $X$ is a reflexive $(x \leq x)$, transitive $(x \leq y \wedge y \leq z \rightarrow x \leq z)$ and anti-symmetric ($x \leq y \wedge y \leq x \rightarrow x = y$) relation.

\subsection{Experiment Storage and Reasoning Structures}\label{ssec:store}

Reuse of experiments requires the storage and retrieval of previously conducted experiments, the queries and requests that led to them, as well as the results, responses and answers generated by them.
To this end, we define an abstract experiment store, while not committing to any specific data storage and processing pipeline or technology.

It is critical that the experiment store is consistent -- while a query is processed, and new requests, results, etc.\ are added, the store is updated; reuse depends on correct recording of the relations between the experiment structures described in the previous section. The reuse algorithm we provide in \Cref{ssec:algo} indeed states that reuse returns the right answer on the condition that the experiment store is consistent (cf.~\Cref{fig:store}).

\begin{definition}[Experiment Store]
    An \emph{experiment store} $S$ for a compatible language structure $LS = \langle \qlang,\erlang,\exscheme \rangle$ is a tuple 
    \[
    S_{LS} = \langle H,F,P,C\rangle_{LS}
    \]
    where $H \subseteq \mathcal{Q}_\qlang \times A_\qlang$ is the set of queries and their answers, 
          $F \subseteq \mathcal{R}_\erlang \times \mathcal{D}_\erlang$ is the set of requests and their responses,  
          $P \subseteq \mathcal{E}_\exscheme \times \mathsf{Res}_\exscheme$ is the set of executed experiment specifications and their result.
          Finally, $C_{LS}$ stores the information that connects the layers with the following signature:
          \begin{align*}
           C \subseteq \quad&\mathcal{Q}_\qlang \times  2^{\mathcal{R}_\erlang} \mbox{(the \decomp relations)}\\
           \cup~& \mathcal{R}_\erlang \times \mathcal{Q}_\qlang \times 2^{\mathcal{E}_\exscheme} \mbox{(the \complete relations)}\\
           \cup~& \mathsf{Res}_\exscheme\times \mathcal{R}_\erlang\times \mathcal{D}_\erlang \mbox{(the \compute relations)}\\
           \cup~& 2^{\mathcal{D}_\erlang} \times A_\qlang \mbox{(the \aggregate relations)}
          \end{align*}
\end{definition}
Note that the storage of the \complete relations also includes a reference to the query corresponding to the experiment request being completed. 


An experiment store can contain arbitrary information, but we are interested in retaining the connections laid out in the previous section.
We use some access functions as shortcuts:
\begin{itemize}
    \item Function $\mathtt{getAnswer}(S_{LS},\quest)$ returns either the unique $\answer$ (in the sense that there is only one \answer for any \quest) such that $(\answer,\quest) \in H$ or $\bot$ if no such answer exists.
    \item Function $\mathtt{getResponses}(S_{LS},\quest)$ returns either the unique set $R$ such that $(\answer,R) \in C$ or $\bot$ if no such set exists.
    \item Function $\mathtt{add}_\mathcal{Q}(S_{LS},\quest,\answer)$ returns a store with the query-answer pair $(\quest, \answer)$ added:
    \(S'_{LS} =  \langle H \cup \{(\quest, \answer)\},F,P,C\rangle_{LS}\)
\end{itemize}

\begin{definition}[Consistency]\label{def:consistent}
Let $LS = \langle \qlang,\erlang,\exscheme \rangle$ be a language structure and $S_{LS} = \langle H,F,P,C\rangle_{LS}$ an experiment store.
We say that $S_{LS}$ is consistent if the following 4 conditions hold.
\begin{align*}
    (i)&H \subseteq \answered \\
    (ii)&F \subseteq \fulfilled \\
    (iii)&P \subseteq \execute \\
    (iv)&C \subseteq \decomp \cup \aggregate \cup \compute \cup \complete
\end{align*}
\end{definition}
An experiment store is \emph{consistent} if it contains only correct connections.
Beyond the language structure, we also define a distance $\dist{\cdot,\cdot}_{A_\qlang}$ on answers and formulate conditions on how sensitive our caching is.

The overall task of our reuse framework thus becomes: given a consistent experiment store $S_{LS}$, and a query $\quest \in \qlang$ to compute an answer $\answer$ such that the functions $\execute$, $\compute$ and $\aggregate$ are invoked as little as possible. As a side-effect, the store $S_{LS}$ will also be updated. 

Formally, we compute $\mathsf{answer}(S_{LS}, \quest) = \langle S'_{LS}, \answer\rangle$ and demand that, for some given \emph{tolerance threshold} $t\in \mathbb{R}^+$:
\begin{enumerate}
    \item $S'_{LS}$ is consistent
    \item $\exists \answer'.~\answered(\quest,\answer') \wedge \dist{\answer,\answer'} \leq t$ -- which expresses that the given answer may be off by $t$, and exact retrieval is the special case of $t=0$.
\end{enumerate}

\subsection{Reuse Mechanisms}\label{ssec:reason}
We formalize four kinds of reuse, of which three are performed on each layer. We briefly illustrate them for when the user layer is given a query \quest, before we give the full definitions.
\begin{description}
    \item[Direct Reuse] If the same question \quest has already been answered with \answer, we merely return the same \answer as before.
    \item[Symbolic Reasoning] If we can deduce an answer \answer for \quest  from the answer of a previous different query $\quest'$ with answer $\answer'$ without invoking \decomp.
    Consider the train running example -- if the train does not stop with a certain configuration within $x$ meters, then it will also surely not stop within $y < x$ meters, in the same conditions. Thus, it is not necessary to perform any experiment for queries with the same configuration and $y < x$. 
    \item[Fuzzy Retrieval] This relies on some distance measure $\dist{\cdot,\cdot}_1$ over queries and a tolerance threshold $t_1$. 
    If there was a previous query $\quest'$ with answer \answer and $\dist{\quest,\quest'}_1 \leq t_1$, then we consider the queries similar enough, and that the answer \answer can be used as an answer for $\quest$ with tolerance $t_1$. For example, if the experiment for the train running example has a minimum sensor resolution of $1$ meter, then two queries only differing by $<1$ meter in the stopping distance will have the same answer. This is essentially a similarity search~\cite{DBLP:series/ads/ZezulaADB06} in the retrieval step for case-based reasoning~\cite{cbr}, with the requirement that the reuse step bounds the error (using a tolerance threshold $t_1$ in our case which is not necessarily the minimum resolution of the observations) as well. Note that this is a very simplistic example, and for different quantities, different tolerance thresholds can and most likely should be set.
    \item[Fuzzy Recomputation] This relies on a different distance measure $\dist{\quest,\quest'}_2$ and corresponding tolerance threshold $t_2$.
    If there was a previous query $\quest'$ with answer $\answer'$ computed from a set of responses $R$ (i.e., $\answer' = \aggregate(\quest',R)$ and $\dist{\quest,\quest'}_2 < t_2$, then we consider the query as so similar, that the answer \answer for \quest can also be computed from $R$, i.e., $\answer = \aggregate(\quest,R)$. This is the case if the queries, e.g., differ in the variable of interest. As this kind of reuse depends on the presence of a lower level, it cannot be implemented in the execution layer, as it is the lowest one.
\end{description}

\begin{figure}
\includegraphics[width=\textwidth]{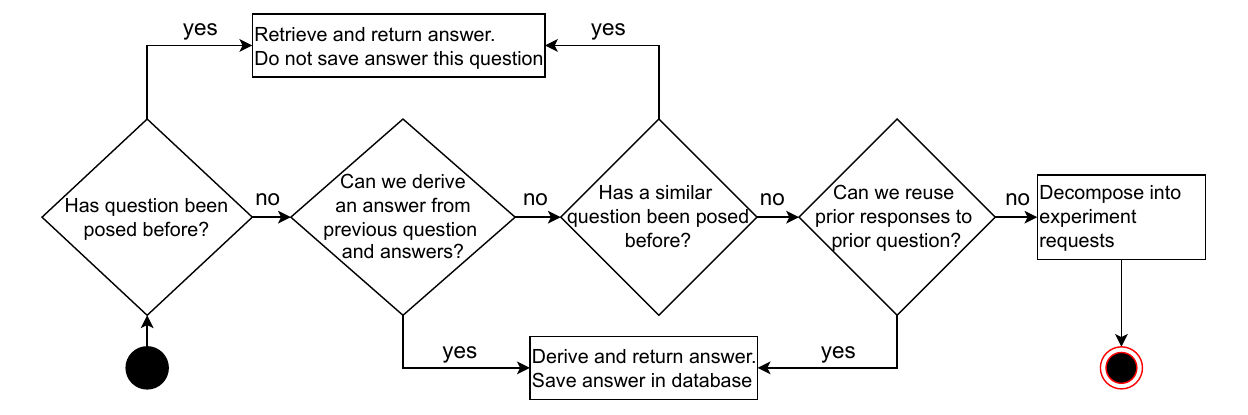}
    \caption{Proposed reuse mechanisms for the user layer in a flow-chart.}
    \label{fig:diamonds}
\end{figure}

\Cref{fig:diamonds} presents a flow-chart showing the composition of these mechanisms in a unified workflow. Direct reuse is executed first, as it is the least computationally heavy, while recomputation is performed last.
Symbolic reasoning and fuzzy retrieval, which are independent of each other and whose computational requirements depend heavily on the technologies used for implementation, can be executed in parallel. While \Cref{fig:diamonds} describes re-use in the user layer, the same mechanisms can be used for the composition and execution layer, with the sole difference that fuzzy recomputation is not possible on the execution layer  -- experiment results are not derived from lower layers and can thus not be recomputed.


\begin{definition}
Let \qlang be a query language. A query reasoning scheme $R_\qlang$ is a tuple
\[
R_\qlang = \left\langle\dist{\cdot,\cdot}_\qlang^{\mathsf{get}}, \dist{\cdot,\cdot}_\qlang^{\mathsf{comp}}, t_\qlang^{\mathsf{get}},t_\qlang^{\mathsf{comp}}\right\rangle
\]
where $\dist{\cdot,\cdot}_\qlang^{\mathsf{get}}, \dist{\cdot,\cdot}_\qlang^{\mathsf{comp}}$ are distances over queries $\mathcal{Q}_\qlang$ and $t_\qlang^{\mathsf{get}},t_\qlang^{\mathsf{comp}} \in \mathbb{R}^+$ are two positive real numbers.
\end{definition}

Distance $\dist{\cdot,\cdot}_\qlang^{\mathsf{get}}$ and threshold $t_\qlang^{\mathsf{get}}$ are used for fuzzy retrieval, while $\dist{\cdot,\cdot}_\qlang^{\mathsf{comp}}$ and threshold $t_\qlang^{\mathsf{comp}}$ are used for fuzzy recomputation. 

To enable fuzzy retrieval and express the second condition for consistency of a language structure with the experiment store, we must be able to bound the error of the answer, based on the distance of the new query to an old query.

\begin{definition}
We say that two queries $\quest_1,\quest_2$ with  $\dist{\quest_1,\quest_2}_\qlang^{\mathsf{get}} < t_\qlang^{\mathsf{get}}$ are $t$-insensitive if the following holds.
\[ \forall \answer_1,\answer_2.~\answered(\quest_1,\answer_1) \wedge \answered(\quest_2,\answer_2) \rightarrow \dist{\answer_1,\answer_2}_\mathsf{A_\qlang} < t
\]

\end{definition}
Intuitively, two queries are insensitive if, under the condition that they have a distance of $t_\qlang^{\mathsf{get}}$, the distance between their responses is bounded by $t$. This means that, if we can accept that the response is slightly off (by $t$), we can reuse answers that are performed for different requests.
We say that \qlang is $t_\qlang^{\mathsf{get}},t$-insensitive if all its queries are pairwise $t_\qlang^{\mathsf{get}},t$-insensitive.

For recomputations, we can have a similar condition. 
\begin{definition}\label{def:aggregate}
Let $\quest_1, \quest_2$ be two queries with $\dist{\quest_1, \quest_2}_\qlang^{\mathsf{comp}} < t_\qlang^{\mathsf{comp}}$.
We say that a function \aggregate is $t$-insensitive, if for any set of responses $R$ such that $\answered(\quest,\aggregate(R, \quest))$, the following holds.
\[
    \dist{\aggregate(R, \quest_1),\aggregate(R, \quest_2)}_\mathsf{A_\qlang} < t
\]
\end{definition}

Similarly,  query languages are equipped with an auxiliary function for reuse.
 \begin{definition}\label{def:justify}
    The function $\mathsf{justify}_\qlang: \mathcal{Q}_\qlang  \times \Delta^a_\qlang \times \mathcal{Q}_\qlang \rightharpoonup \Delta^a_\qlang$ takes as input a query $q_1$, an answer $a$ for it, and a second query $q_2$. It computes an answer $a_2$ for $q_2$, i.e., the following holds.
    \[
    \big(\mathsf{justify}_\qlang(q_1,a_1,q_2) = a_2 \wedge \answered(q_1,a_1)\big) \Rightarrow \answered(q_2,a_2)
    \]
    Note that the function is partial.
 \end{definition}

 Let us return to the running example to illustrate symbolic reasoning.

\begin{example}
We know that if a train is safe in a certain situation, then it will also be safe in a less dire situation (lower initial velocity, and/or higher mass, braking force, or friction coefficient, see proportionalities in \Cref{sec:example}). Similarly, we know that if a train is not safe in a given situation, then it will also not be safe in a more dire situation. Formally, we define this as follows:
\begin{align*}
    &\mathsf{justify}_{\qlang^\mathsf{Train}_1}(q_1, \mathtt{true}, q_2) = \mathtt{true} \\
    \iff &(q_1(\mathit{v}) \geq q_2(\mathit{v})) \wedge (q_1(\mathit{m}) \leq q_2(\mathit{m})) \wedge 
    (q_1(\mathit{F_B}) \leq q_2(\mathit{F_B})) \\
    &\wedge (q_1(\mu) \leq q_2(\mu)) \wedge (q_1(\theta) \leq q_2(\theta))\\ 
    \vspace{8pt} \\ 
    &\mathsf{justify}_{\qlang^\mathsf{Train}_1}(q_1, \mathtt{false}, q_2) = \mathtt{false} \\
    \iff &(q_1(\mathit{v}) \leq q_2(\mathit{v})) \wedge (q_1(\mathit{m}) \geq q_2(\mathit{m})) \wedge 
    (q_1(\mathit{F_B}) \geq q_2(\mathit{F_B})) \\
    &\wedge (q_1(\mu) \geq q_2(\mu)) \wedge (q_1(\theta) \geq q_2(\theta))
\end{align*}

\end{example}

\begin{figure}
\begin{lstlisting}[language=pseudo, caption={Pseudo-code algorithm for reuse on the user layer.}, label={fig:algo-user}]
$\text{\textbf{function} }$reuse-user
$\text{\textbf{Input}:~~A query}$ $\quest\text{ and a store}$ store
$\text{\textbf{Output}: \!\!\!A pair of a store and either an answer, or } \bot$

//implements direct reuse
if($\mathsf{getAnswer}($store$,\quest) = \answer$) return $\langle $store$, \answer\rangle$

//implements symbolic reasoning
if($\exists \quest'.\mathsf{getAnswer}($store$,\quest') = \answer' \wedge \justify_Q(\quest,\quest',\answer') = \answer$)
  return $\langle \mathsf{add}_Q($store$,\quest,\answer), \answer\rangle$ $\label{l:cons1}$
end

//implements fuzzy retrieval
if($\exists \quest'.\mathsf{getAnswer}($store$,\quest') = \answer \wedge |\quest,\quest'|^\mathsf{get}_\qlang < t^\mathsf{get}_\qlang$)
  return $\langle$store$, \answer\rangle$ $\label{l:cons3}$
end

//implements fuzzy recomputation
if($\exists \quest'.\mathsf{getResponses}($store$,\quest') \neq \bot \wedge~|\quest,\quest'|^\mathsf{comp}_\qlang < t^\mathsf{comp}_\qlang$)
  ans := $\aggregate(\mathsf{getResponses}($store$,\quest'),\quest)$
  return $\langle\mathsf{add}_Q($store$,\quest,$ans), ans$\rangle$ $\label{l:cons2}$
end
return $\langle$store$, \bot\rangle$
\end{lstlisting}
\end{figure}

Note that the above formulation also subsumes the case where $q_1 = q_2$, i.e. direct reuse. Next, we examine how insensitivity, compatibility and reasoning structures enable reuse on the user layer.

\Cref{fig:algo-user} shows the reuse algorithm in pseudo-code. 
It takes as input a query \quest and a store $\mathtt{store}$, and outputs either an updated store and an answer, or the symbol $\bot$ that denotes that no reuse is possible.  
The first step (line 6) is direct reuse, where the store is not updated (as it already contains the query-answer pair). The second step (line 9) is symbolic reasoning. It checks whether a query $\quest'$ exists, with an answer $\answer'$ that can justify an answer for $\quest$. This does update the store, as the new answer was not stored yet. The third step (line 14) is fuzzy retrieval.
Here, the store is not updated: the answer \answer is originally for $\quest'$, adding it as an answer for $\quest$ would increase the distance where it can be reused beyond the original semantics.
Finally, the last step (line 19) is fuzzy recomputation. Again, the results are stored, as they are computed specifically for the new query.

We can phrase a correctness notion for this reuse on the user layer.
\begin{lemma}\label{lem:user}
    Let $\quest$ be a query, $\mathsf{store}_{LS}$ a consistent store for a $t$-insensitive language structure, and $\langle \mathsf{store}',\answer \rangle = \mathtt{reuse-user}(\mathsf{store}, \quest)$, where $\mathtt{reuse-user}$ is the function described in \Cref{fig:algo-user}. 
    (1) The resulting store $\mathsf{store}'$ is consistent, i.e., \(\mathsf{consistent}(\mathsf{store'})\)
    and (2) the returned answer is correct with sensitivity $t$.
    $\exists \answer'.~\answered(\quest,\answer') \wedge \dist{\answer,\answer'} \leq t$.
\begin{proof}
Let us first examine (1). First, we observe that only the $H$ component of the store is ever modified, so it suffices that we preserve condition (i) of \cref{def:consistent}. There are two modifications. In \cref{l:cons1}, the new answer $\answer$ is computed from an existing answer $\answer'$ and a corresponding, prior query. By the definition of $\justify$ (\cref{def:justify}), this preserves the condition on \answered. In \cref{l:cons2}, it is established for the retrieved answer by compatibility of the language structure. Note that $\mathtt{getResponses}$ must be from a prior decomposition, thus fulfilling the condition on $F$ in \cref{def:consistent}.

As for (2), the relevant retrieval is in \cref{l:cons3}, where the return answer is correct due to the assumed insensitivity of \aggregate (\cref{def:aggregate}).\qedhere
\end{proof}
\end{lemma}

\subsection{Reuse Algorithm and Soundness}\label{ssec:algo}

The decomposition and execution layers have the very same ideas and structure as the user layer, with the sole exception that the reasoning scheme for experiment specification (i.e. the execution layer) lacks the second distance and threshold, as no fuzzy recomputation is possible. We, thus, define them without extensive elaboration.
\begin{definition}
Let \erlang be an experiment request language and \exscheme be an experiment specification language. A request reasoning scheme $R_\erlang$ is a tuple
\[
R_\erlang = \left\langle\dist{\cdot,\cdot}_\erlang^{\mathsf{get}}, \dist{\cdot,\cdot}_\erlang^{\mathsf{comp}}, t_\erlang^{\mathsf{get}},t_\erlang^{\mathsf{comp}}\right\rangle
\]
where $\dist{\cdot,\cdot}_\erlang^{\mathsf{get}}, \dist{\cdot,\cdot}_\erlang^{\mathsf{comp}}$ are distances over requests $\mathcal{R}_\erlang$ and $t_\erlang^{\mathsf{get}},t_\erlang^{\mathsf{comp}} \in \mathbb{R}^+$ are two positive numbers.
 A specification reasoning scheme $R_\exscheme$ is a tuple
\[
R_\exscheme = \left\langle\dist{\cdot,\cdot}_\exscheme^{\mathsf{get}}, t_\exscheme^{\mathsf{get}}  \right\rangle
\]
where $\dist{\cdot,\cdot}_\exscheme^{\mathsf{get}}, \dist{\cdot,\cdot}_\exscheme^{\mathsf{comp}}$ are distances over queries $\mathcal{R}_\exscheme$ and $t_\exscheme^{\mathsf{get}},t_\exscheme^{\mathsf{comp}} \in \mathbb{R}^+$ are two positive numbers.
\end{definition}


The design of insensitive structures is highly domain- and application-specific, but we retain the notion of $t$-sensitivity. The algorithms are, for brevity's sake, in the appendix.

\begin{definition}
We say that two requests $\request_1,\request_2$ with  $\dist{\request_1,\request_2}_\erlang^{\mathsf{get}} < t_\erlang^{\mathsf{get}}$ are $t$-insensitive if the following holds.
\[ \forall \response_1,\response_2.~\fulfilled(\request_1,\response_2) \wedge \fulfilled(\request_2,\response_2) \rightarrow \dist{\response_1,\response_2}_\mathsf{R_\erlang} < t
\]
\end{definition}

\begin{example}
    In the train running example, as the train is starting to break immediately, we can bound the additional distance it needs to break as follows.
    Let $r_1,r_2$ be two requests. For the sake of brevity, we define a simple exemplary distance between two requests as the absolute difference between their initial velocities, in case all other variables are equal, or $\infty$ otherwise.
    A corresponding exemplary distance between two responses in the train example is the absolute difference between the their values, in this case the stopping distance.
    
    Having the domain knowledge that the train follows the simple kinematic equation $\mathit{pos}(t) = v_0*t+0.5a*t^2$ with error $<v_0*0.1$, we can say that
    if the requests have distance $x$, then the responses have at most distance $0.1x$.
    \begin{align*}
        |r_1,r_2| < x \rightarrow |rs_1,rs_2| < x*0.1
    \end{align*}
\end{example}

We can now describe the whole workflow for re-use across all three layers in the algorithm in \Cref{fig:algo-total}.
It first attempts reuse on the user layer (line 5).
If no reuse is possible, it decomposes the query into requests (line 7) and attempts reuse for each of those (line 12). Those for which no reuse is possible are then subsequently completed into specifications (line 15). Finally, reuse is attempted for the experiment specifications (line 19).

\begin{figure}
\begin{lstlisting}[language=pseudo, caption={Pseudo-code algorithm for overall reuse across the three layers}, label={fig:algo-total}]
$\text{\textbf{function} }$process
$\text{\textbf{Input}: A query}$ $\quest$ and a store store
$\text{\textbf{Output}: \!\!\!A pair of a store and an answer}$
\end{lstlisting}
\begin{minipage}{.5\textwidth}
\begin{lstlisting}[language=pseudo,firstnumber=4]
$\langle$st, an$\rangle$ := reuse-user($\quest$,store)
if(an $\neq$ $\bot$) return $\langle$st, an$\rangle$

R := $\decomp(\quest)$ 
E := $\emptyset$            
RSP := $\emptyset$          
foreach $\request$ $\in$ R do
   $\langle$st, rs$\rangle$ := 
     reuse-composition($\request$,store)
   store := st
   if(rs $\neq$ $\bot$) RSP := RSP $\cup$ $\{$rs$\}$
   else       E := E $\cup$ $\complete($rs$)$ 
end
\end{lstlisting}
\end{minipage}
\begin{minipage}{.5\textwidth}
\begin{lstlisting}[language=pseudo, firstnumber=17]
foreach $\experiment$ $\in$ E do
   $\langle$st, er$\rangle$ := 
     reuse-experiment($\experiment$,store)
   store := st
   if(er $\neq$ $\bot$) rslt := er
   else       rslt := $\execute($e$)$ 
   RSP := RSP $\cup$ $\compute(\experiment,$rslt$)$ 
end

ans := $\aggregate$($\quest$,RSP) 
return $\langle $store$, $ans$ \rangle$
\end{lstlisting}
\end{minipage}
\end{figure}

\begin{theorem}[Soundness]
    Let $\quest$ be a query, $\mathsf{store}_{LS}$ a consistent store for a $t$-insensitive language structure, and $\langle \mathsf{store}',\answer \rangle = \mathtt{process}(\mathsf{store}, \quest)$. 
    The resulting store $\mathsf{store}'$ is consistent, i.e., \(\mathsf{consistent}(\mathsf{store'})\)
    and the returned answer is correct with sensitivity $t$.
    $\exists \answer'.~\answered(\quest,\answer') \wedge \dist{\answer,\answer'} \leq t$.
\end{theorem}
\begin{proof}
For brevity, we avoid stating \Cref{lem:user} for \texttt{reuse-decomposition} and \texttt{reuse-execution} (described in \ref{sec:appendix-reuse-layers-code}), for which the respective proof is analogous. Overall correctness follows directly from these lemmas and the simple observation that all data stored and all answers returned originate in the reuse function or a direct computation, for which consistency trivially holds.
\end{proof}
\section{Managing Experiments and Experiment Data}\label{sec:architecture}

The proposed mechanisms rely on the correct interplay of numerous components, each of which must scale on its own, has a unique life-cycle and requires specific technical knowledge for development and maintenance. This is particularly critical as experiments and reasoning may require monitoring a process for a longer time.
To enable the scalability of our system, we propose a service-oriented architecture that handles experiments and their reuse.
We emphasize several aspects that must be considered by the architecture.

First, we observe that the experiment store contains two components that require fundamentally different data storage technologies: Experiment data is high-volume, numerical time series data with semantic tags, while the remaining structure is symbolic data that refers to experiment data. 
Furthermore, the experiment data storage must be accessed by the experiments, while the rest is stored for reuse and reasoning. We therefore use two different stores to emphasize this difference, but they can easily be combined using a federated database.
    
Second, we must consider that experiments run extensive simulations or physical setups, thus the component handling them must be decoupled from the remaining workflow for responsiveness. This points towards an event-driven architecture, which we assume in the following.
    
Finally, if the system is to be used by several units of an organization, or several different organizations, in case it is offered as a service, reuse must be controlled so no confidential data is exposed through the system.
    
\Cref{fig:Architecture} shows an overview of our architecture for two input query languages and two different kinds of experiments. 
In the following, we describe each component and which parts of the formal description it uses or addresses.

\begin{figure}[ht]
    \centering
    \includegraphics[width=\textwidth]{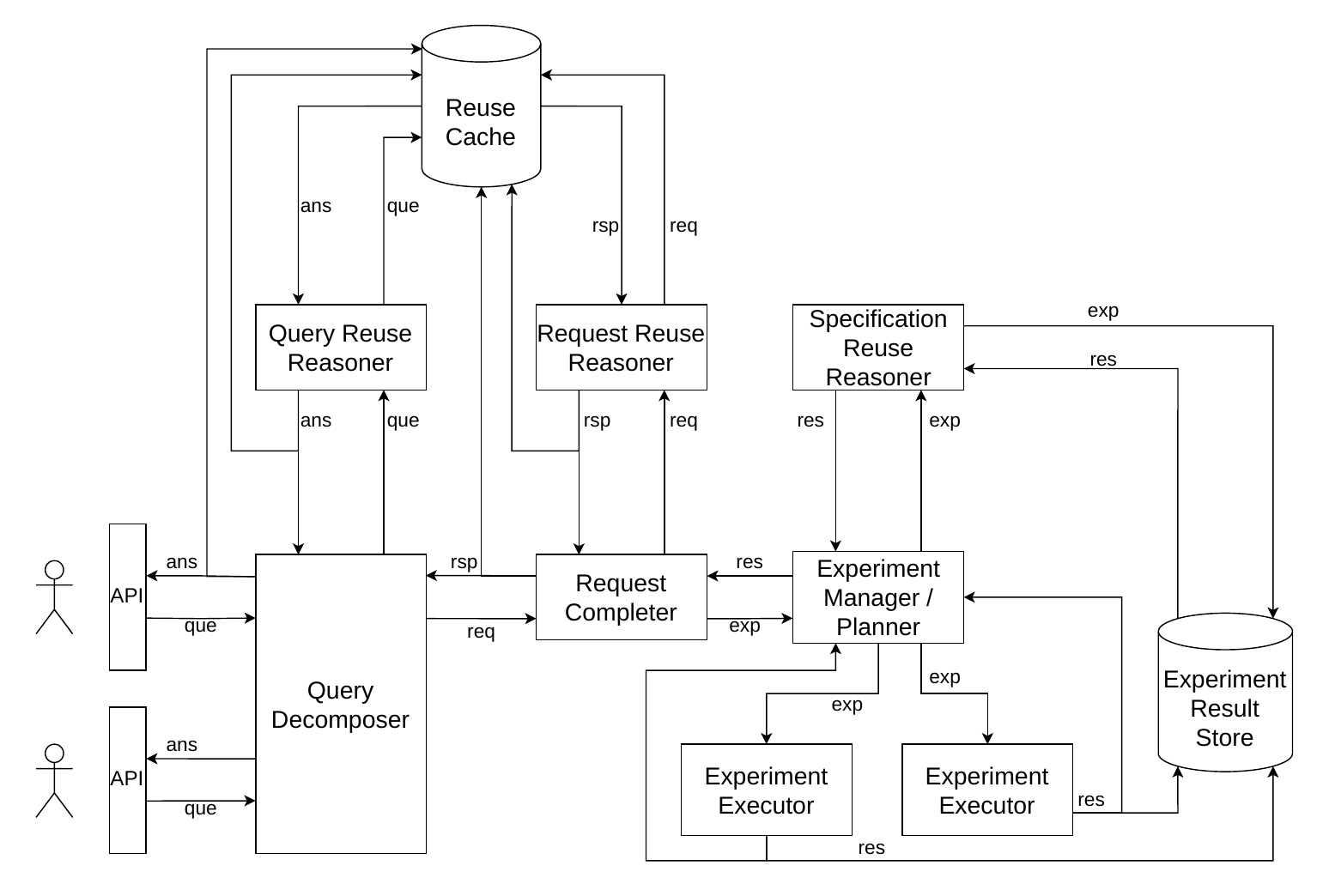}
    \caption{Sketch of the architecture to realize the proposed framework}
    \label{fig:Architecture}
\end{figure}

\paragraph{API and Query Decomposer} The highest-level interface to the system is via an Application Programming Interface (API), which enables user agents to pose queries to a set of predefined query languages. The APIs possibly realize a set of query languages \qlang or just a single \qlang.  Each user query is communicated to the Query Decomposer component.
The query decomposer, for each query:
\begin{enumerate}
\item First passes the query to the Query Reuse Reasoner component to realize the \texttt{reuse-user} function. If a corresponding reusable answer is found and sent back from the Reasoner component, it communicates the answer to the user via the API.
\item If no reuse is possible at the user-level, it realizes the \decomp function to decompose the query into its corresponding experiment request/s and passes them to the Request Completer. Once the corresponding responses arrive, it realizes the \aggregate function to aggregate the responses into an answer that is sent to the user via the API. In such a scenario, the Query Decomposer also communicates with the Reuse cache to cache the corresponding answer.
\end{enumerate} 

\paragraph{Query Reuse Reasoner and Request Reuse Reasoner} The reuse on the user layer and the decomposition layer is realized by the Query Reuse Reasoner and the Request Reuse Reasoner, which respectively implement the \texttt{reuse-user} and \texttt{reuse-decompose} functions. If a similar query/request is used in the reply, or an answer/response is derived from previous queries/requests, then the Reasoner components are also responsible for updating the cache to store the association between the current query/request and the corresponding (derived or nearby) answer/response in the cache (which is shown as the outward arrow from the Reasoner blocks providing an answer \answer to the Reuse Cache).

\paragraph{Request Completer} The request completer, for each request:
\begin{enumerate}
\item First passes the request to the Request Reuse Reasoner component to realize the \texttt{reuse-decomposition} function. If a corresponding reusable response is found and sent back from the Reasoner component, it communicates the response to the Query Decomposer.
\item If no reuse is possible at the decomposition-level, it realizes the \complete function to complete the request into its corresponding experiment specification/s and passes them to the Experiment Manager / Planner component. Once the corresponding results arrive, it realizes the \compute function to compute the property of interest that is sent to the Query Decomposer component. In such a scenario, the Request Completer also communicates with the Reuse cache to cache the corresponding response.
\end{enumerate}

\paragraph{Experiment Manager / Planner}
This component realizes \exscheme and uses $\leq_\exscheme$ to plan, schedule and optimize the experiment specifications stemming from one query. However, before performing the above tasks, the Experiment Manager communicates the experiment specification to the Specification Reuse Reasoner to realize the \texttt{reuse-execution} also optimizes the execution of experiments stemming from different queries. As such, there is only one planner in the system. The planner does not realize \execute, but sends an experiment specification to an experiment executor for execution.

\paragraph{Specification Reuse Reasoner}
Similar to the previous reasoning components, the Specification Reuse Reasoner realizes the \texttt{reuse-execution} function.

\paragraph{Experiment Executor}

An Experiment Executor executes the experiment corresponding for each experiment specification passed to it, i.e., it realizes the \execute function. The resulting data is stored in the Experiment Result Store, and the Experiment Planner is notified. There are two Experiment Executors shown in \Cref{fig:Architecture}, each responsible for running a different kind of experiment.

\paragraph{Storage}
There are two storages that together realize an experiment store $S = \langle H,F,P,C \rangle$\footnote{They realize one experiment store for all query languages used, so that only the $H$ components is actually duplicated.}. The \emph{Experiment Result Store} stores the results of the experiments, i.e., it realizes $P$ and is updated by the experiment manager. The \emph{Reuse Cache} realizes the remaining components $H,F,C$, where $H$ is updated by the Query Reuse Reasoner, $F$ by the Request Reuse Reasoner and $C$ by all components. Note that the \emph{Reuse Cache} needs to be implemented with certain TTL (time-to-live) policies that ensure outdated information in the cache is purged a.k.a.~cleaned continually.

\section{Evaluation}\label{sec:evaluation}


Let us return to the Thermal Management System (TMS) case-study described in~\Cref{sec:overview} and instantiate our framework for this realistic industrial scenario to evaluate whether it can indeed be used to express domain knowledge for reuse, and enable reuse to reduce experiment executions. First, recall that the design of the TMS fixes the layout of the necessary components, i.e., the required components and their interfaces, and experiments must be used to determine the parameters of the components.

The queries posed by the engineers in the context of early engineering are (1) under-specified yet constrained TMS designs, i.e.,~the experiment requests and (2) fully specified TMS designs where experiment specifications are scripts that load and configure simulation units to perform the experiments. Reuse is possible, and indeed required, to perform as few experiments as possible.

\subsection{Formalizing case-study as a tripartite structure}
In the following, we formally instantiate our framework for this case study. 
We consider layouts determined by two components: a battery and its environment. The aim of the engineer is to derive a Pareto optimal solution (a.k.a.~Pareto front) that guides further decision-making. Domain knowledge must encode when an experiment will, in no case, move the Pareto front.

\begin{definition}
    A \emph{layout} is a pair $\langle B,\mathsf{Env} \rangle$ of two maps $B, \mathsf{Env}$ with the following signatures.
\begin{align*}
    B &: \{\mathsf{Voltage},  \mathsf{MaxTorque}, \mathsf{InternalRes}\} \rightarrow \mathbb{R} \cup \{\ast\}\\
    \mathsf{Env} &: \{\mathsf{stim}\} \rightarrow \mathsf{Stim}
\end{align*}
\end{definition}

The map $B$ models the parameters of the battery and the map $\mathsf{Env}$ models the parameters of its environment. We use the special symbol $\ast$ to denote that a parameter is under-specified and yet to be determined. $\mathsf{Stim}$ denotes the set of stimulation profiles for a battery, i.e., descriptions of a usage profile.

A layout is \emph{fully determined} if no variable is mapped to $\ast$. Otherwise, it is \emph{under-specified}.
A \emph{constrained TMS} (CTMS) is a pair of a TMS and a set of inequalities over all variables that map to $\ast$.

Let $\mathsf{Ineq}$ be the set of satisfiable inequalities over the parameters of layouts.
If a layout $\langle B,\mathsf{Env} \rangle$ satisfies an inequality $\mathit{ieq}$, then we write $\langle B,\mathsf{Env} \rangle \models \mathit{ieq}$.
We furthermore denote with $\mathsf{fulldet}(\langle B,\mathsf{Env} \rangle)$ the set of layouts that are (1) fully determined and (2) $\langle B,\mathsf{Env} \rangle \models \mathit{ieq}$, $\forall \mathit{ieq} \in \mathsf{Ineq}$.


\begin{example}
    The query language $\qlang^\mathsf{TMS}$ for the TMS case-study is defined by the tuple $\langle V^\mathsf{TMS}_\qlang, \mathcal{L}^\mathsf{TMS}_\qlang, \Delta^\mathsf{TMS}_\qlang, A^\mathsf{TMS}_\qlang\rangle$ with the following definitions.
    \begin{align*}
        V^\mathsf{TMS}_\qlang &=  \mathit{InTMS} \cup \{\mathsf{Constr}\}\\
        \mathcal{L}^\mathsf{TMS}_\qlang &= \{ V^\mathsf{TMS}_\qlang \}\\
        \Delta^\mathsf{TMS}_\qlang &=  \big(\mathit{InTMS}\rightarrow (\mathbb{R} \cup \{\ast\})\big) \cup \big(
        \mathsf{Constr} \rightarrow \mathsf{Ineq}\big) \\ 
        A^\mathsf{TMS}_\qlang &= \mathcal{P}\big(\mathit{InTMS} \cup \mathit{OutTMS} \rightarrow \mathbb{R}\big)
    \end{align*}
\end{example}

$\mathit{InTMS}$ denotes the subset of the variables that are layout parameters of the battery and its environment:
\[
\mathit{InTMS} = \mathbf{dom} B \cup \mathbf{dom} \mathsf{Env} = \{\mathsf{Voltage},  \mathsf{MaxPower}, \mathsf{InternalRes}, \mathsf{stim} \}
\]

$\mathsf{Constr}$ represents the set of constraints. 
The query scheme $\mathcal{L}^\mathsf{TMS}_\qlang$ ensures that all variables must be specified. Function $\Delta^\mathsf{TMS}_\qlang$ defines the possible values for the variables:
The input variables, stimulation and constraint variable are mapped as in the layout.

Simulation of the battery system with a given configuration yields two outputs: 
TotalBatteryLosses ($\mathsf{TBL}$), and Minimal State of Charge ($\mathsf{SoC}$). We denote these variables with $\mathit{OutTMS}$. 

The answer returned to the user is a set of tuples, each representing a fully determined battery configuration and the outputs from simulation of that configuration, with one entry per input and output variable. The answer is modelled as a Pareto front $P$ stemming from the optimization problem that $\mathsf{TBL}$ must be minimized, while $\mathsf{SoC}$ must be maximized.

\begin{example}
    The request language $\erlang^\mathsf{TMS}$ for TMS is defined by the following tuple $\langle V_\erlang^\mathsf{TMS}, I_\erlang^\mathsf{TMS}, \Delta_\erlang^\mathsf{TMS}\rangle$ with the following definitions.
    \begin{align*}
         V_\erlang^\mathsf{TMS} &= \mathit{InTMS}\\ 
         I_\erlang^\mathsf{TMS} &=\mathit{InTMS} \cup \mathit{OutTMS} \rightarrow \mathbb{R}\\
         \Delta_\erlang^\mathsf{TMS} &= \mathit{InTMS} \rightarrow \mathbb{R}
    \end{align*}
\end{example}

\begin{example}
    The specification language $\exscheme^\mathsf{TMS}$ is defined by the following.
    \[
    \langle\experiments^\mathsf{TMS}, \leq_\exscheme^\mathsf{TMS}, \mathsf{Res}_\exscheme^\mathsf{TMS}
    \rangle 
    \]
    where $\experiments^\mathsf{TMS} = \mathsf{Load}(\overline{x}) \cup \mathsf{Execute}$ is the set of all pairs of FMUs loaded for the experiments with parameters $\overline{x}$ and the set of all master algorithms $\mathsf{Execute}$. The order defines that each element of $\mathsf{Load}(\overline{x})$ is to be executed before any element of $\mathsf{Execute}$. Finally:
    \begin{align*}
        \mathsf{Res}_\exscheme^\mathsf{TMS} &= \mathbb{R}^+ \rightarrow\mathbb{R} \times \mathbb{R}
    \end{align*}
\end{example}

The specification language is not based on the request language -- it specifies the simulation experiment itself, which in our case means that it describes two simulation units and their co-simulation master algorithm. The results are time-series over both output variables. In our example, the simulation experiment has two steps: first, load the FMUs and initialize the algorithm, second, execute the algorithm and store it.

We furthermore need functions to translate between the different layers.
\begin{description}
    \item[Decomposition] Decomposition translates a query into a set of requests. In our case, it generates a set of all fully instantiated layouts that have at least a minimal distance to each other\footnote{Note that we use the reals to model inputs, but even when using a finite set, the set of all fully instantiated layouts would not be practical. While reuse would solve this, it is unnecessary to generate these requests in the first place.}. Let $d > 0$ denote this distance.
    \begin{align*}
    &\decomp\big(\langle B,\mathsf{Env} \rangle \cup \{\mathsf{Constr} \mapsto \mathit{ieq}, \mathsf{PoI} \mapsto P\}\big) \\
    =& \Big\{\big(\langle B',\mathsf{Env}' \rangle \cup \{\mathsf{PoI} \mapsto P\}\big) ~|~ \langle B',\mathsf{Env}' \rangle \models \mathit{ieq}, \langle B',\mathsf{Env}' \rangle \in \mathsf{compl}(\langle B,\mathsf{Env} \rangle)\Big\}
    \end{align*}
    Such that \[\forall r,r' \in \decomp\big(\langle B,\mathsf{Env} \rangle \cup \{\mathsf{Constr} \mapsto \mathit{ieq}, \mathsf{PoI} \mapsto P\}\big).~|r-r'| \doteq d\]
    \item[Aggregation]
    Aggregation is computing the Pareto front from all responses. This is standard, and we abstract it into a function $\mathsf{opt}$.
    \begin{align*}
    &(\langle B,\mathsf{Env} \rangle, P) \\
    = & \mathsf{opt}(\{\response \mid \forall r,r' \in \decomp\big(\langle B,\mathsf{Env} \rangle \cup \{\mathsf{Constr} \mapsto \mathit{ieq}, \mathsf{PoI} \mapsto P\}\big)\})
    \end{align*}
    \item[Completion] Completion maps the request to a specification -- in case of the battery, this refers to a script that runs an FMU configured with the pre-defined standard drive cycle for the battery. 
    
    \item[Computation] Computation is a function that computes the metrics $\mathsf{TBL}$ and $\mathsf{SoC}$, from the traces of the executed simulation. It simply returns the last values of the respective traces.
\end{description}

In this case study, we identify the following reuse mechanisms:

The first is reuse on the user level -- this can be the case if the user poses a query that is close to one asked before. Maybe the user is not aware of the first query, because it was posed by another user, or because they are underestimating the distance between the queries.
 
The second is distance-based reuse on the request level. As the request is not obvious to the user, who only sees the query, several queries may result in the same or similar requests. We demonstrate reuse on the decomposition layer using the following two requests $\request_1$ and $\request_2$ as examples. 
\begin{align*}
\request_1 &= \{\mathsf{Voltage} \mapsto V_1, \mathsf{MaxTorque} \mapsto T_1, \mathsf{InternalRes} \mapsto R_1, \mathsf{PoI} \mapsto \{ \mathsf{SoC}_1\}\}\\
\request_2 &= \{\mathsf{Voltage} \mapsto V_2, \mathsf{MaxTorque} \mapsto T_2, \mathsf{InternalRes} \mapsto R_2, \mathsf{PoI} \mapsto \{ \mathsf{SoC}_2\}\}
\end{align*}
    
Formally, we instantiate fuzzy retrieval as follows. The distance between two requests is defined as a Boolean function that compares each of the first three parameters ($\mathsf{Voltage}, \mathsf{MaxTorque}, \mathsf{InternalRes}$) with their individual distance thresholds. This is necessary because, for example, a difference of 5 V in the $\mathsf{Voltage}$ may be small, but a difference with the same magnitude (i.e. 5 \ohm) in the $\mathsf{InternalRes}$ of the battery is more significant.

\begin{align*}
    \dist{\request_1,\request_2}_\erlang^{\mathsf{get}} &=
    \begin{cases}
        0 & \text{if } \quad |V_1 - V_2| \leq t_\mathsf{Voltage} \\
        & \wedge \quad |T_1 - T_2| \leq t_\mathsf{MaxTorque}  \\
        & \wedge \quad |R_1 - R_2| \leq t_\mathsf{InternalRes} \\
        \infty & \text{otherwise}
    \end{cases}
\end{align*}

Note that there are many ways to account for multi-variate fuzzy retrieval -- another approach (that we mention here, but do not use in the evaluation) is to define the distance between two requests as the sum of normalized differences between all variables, with a threshold roughly corresponding to the precision of the experiments.

\begin{align*}
    \dist{\request_1,\request_2}_\erlang^{\mathsf{get}} &= \frac{|V_1 - V_2|}{V_1+V_2} + \frac{|T_1 - T_2|}{T_1 + T_2} + \frac{|R_1 - R_2|}{R_1+R_2}
\end{align*}
 
The third reuse mechanism, symbolic reuse, is defined as follows. If one can reason about the optimality of a new response based on existing ones, we can avoid executing the request. We have identified the following symbolic reuse for the case-study:

\begin{itemize}

\item The final $\mathsf{SoC}$ (State of Charge) of the battery is directly proportional to $\mathsf{Voltage}$ and $\mathsf{MaxTorque}$, and inversely proportional to $\mathsf{InternalRes}$. Recall that a battery layout is unstable if the lowest state of charge at any point in the drive cycle is lower than 50\%. If a previous experiment is found in the store which is more stable in parameters (i.e. based on the previously described inequalities, its $\mathsf{SoC}$ will be higher than the current request) but its simulation trace indicates instability, then it can be inferred that the current requested experiment is also an unstable combination of parameters and hence need not be executed.
\end{itemize}
 
Formally, we define the symbolic reuse intuition as follows. Note that a False answer is interpreted by the optimization algorithm as a point in the Pareto-optimization problem space to be skipped (since it is evident that it is going to be dominated). 
    \begin{align*}
        &\mathsf{justify}_{\erlang^\mathsf{TSM}}(\request_1,a,\request_2) = False    \\
        &\Leftrightarrow \exists \request_2 \mid \\\
        &\big(V_1 \leq V_2 \wedge T_1 \leq T_2 \wedge R_1 \geq R_2 \wedge \mathsf{SoC}_2 < 50\big) 
    \end{align*}
    
The fourth and final reuse mechanism is to reuse experiment data on the lowest layer: If one request was used to retrieve one PoI, then the generated data/traces still contain the time series for the other PoI, and can thus be reused without executing the experiment again. Formally, we instantiate fuzzy recomputation as follows.
     
    \begin{align*}
        \dist{\request_1,\request_2}_\erlang^{\mathsf{comp}} &= 
        \begin{cases}
        0 & \text{if}~V_1 = V_2, T_1 = T_2, R_1 = R_2\\
        \infty & \text{otherwise}\\
        \end{cases} \\
        t_{\erlang}^{\mathsf{comp}} = 1
    \end{align*}




\subsection{Qualitative Evaluation}
To evaluate our experiment reuse framework, we prototypically implemented both the case-study (i.e., the battery layout) and the running example (i.e., the train braking behavior), to answer the following research questions.
\begin{description}
    \item[RQ1] Does the proposed framework reduce the number of experiments performed with an increase in the number of queries it answers? 
    \item[RQ2] Does experiment reuse increase efficiency of answering a query w.r.t time and memory?
    \item[RQ3] Does the framework support explicating the domain knowledge needed for reuse of experiments?
\end{description}
 
\subsubsection{Prototypical Implementation}

In the following, we present the experiment setup and the reusable artifacts supplementing this paper, which are available in the form of pre-compiled Docker images under \href{https://doi.org/10.5281/zenodo.15528008}{https://doi.org/10.5281/zenodo.15528008}.
We prototyped the framework using a number of off-the-shelf software technologies:

\begin{itemize}
\item PostgreSQL \footnote{\url{https://www.postgresql.org}} with the TimescaleDB\footnote{\url{https://github.com/timescale/timescaledb}} and PGVector\footnote{\url{https://github.com/pgvector/pgvector}} extensions. The TimescaleDB extension was used to efficiently store and query time-series data, while the PGVector extension enabled efficient vector-based similarity search for a number of distance functions (for fuzzy reuse).
\item Redis for fast caching of the requests and responses (which together with the PostgreSQL database formed the Experiment Store $S = \langle H,F,P,C \rangle$)
\item Python3 with FastAPI, Uvicorn, and Pydantic to deploy the backend API along with request and response model validation. The code realizing the various functions like \decomp, \aggregate, \compute, \complete, and \execute was also implemented in separate Python modules. FMPy was used to simulate the FMU provided by our industrial partner. 
\item Node.js with React and TypeScript to create and deploy the Graphical User Interface as shown in \Cref{fig:gui}.
\item Docker to containerize all individual components for modularity, maintainability, and reusability.
\end{itemize}

\begin{figure}
    \centering
    \includegraphics[width=\linewidth]{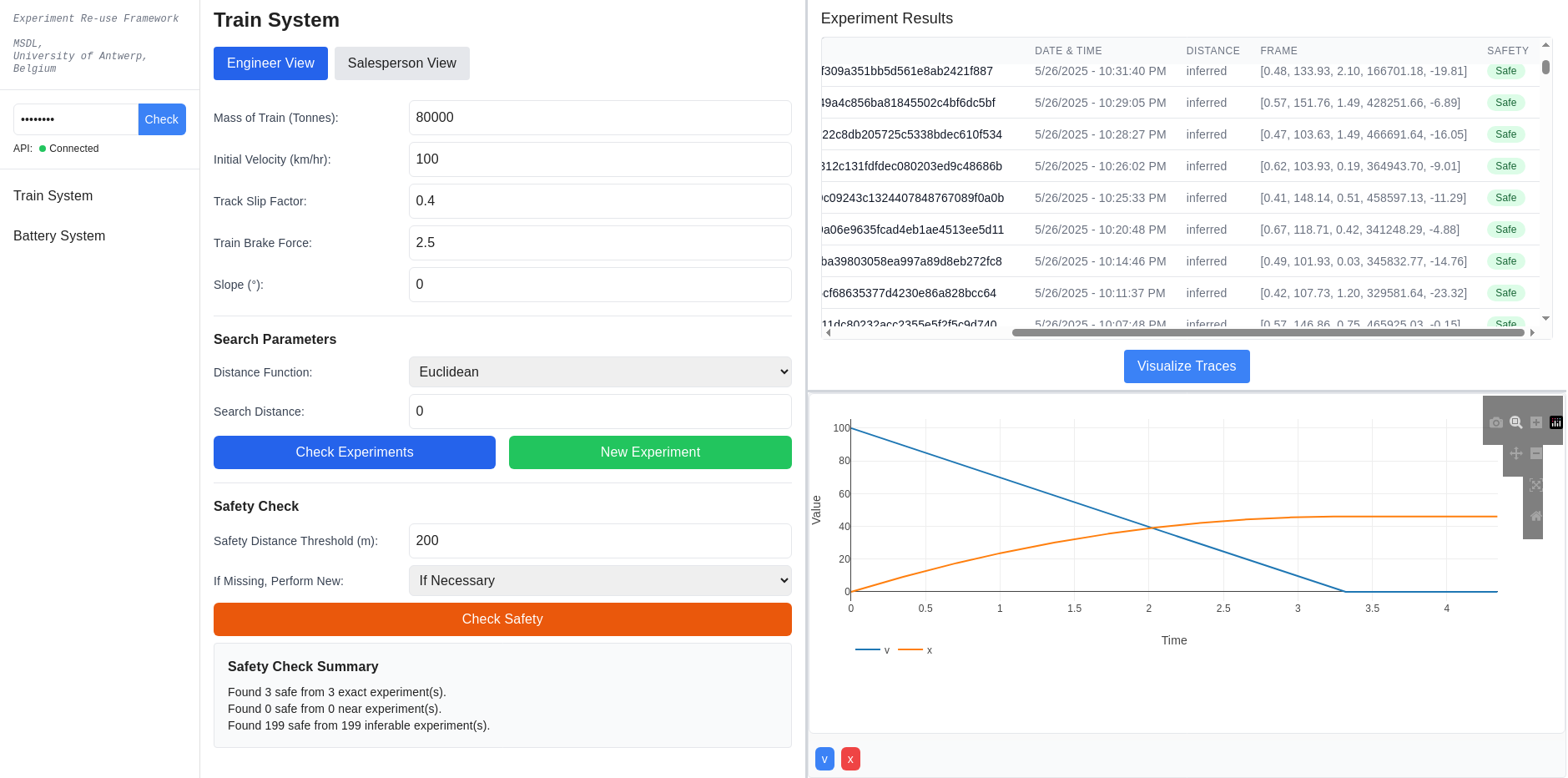}
    \caption{Web-based Graphical User Interface in the implementation}
    \label{fig:gui}
\end{figure}

\subsubsection{Evaluation Setup}

To answer each question, we performed different evaluations on our prototypical implementation as described below:

\paragraph{\textbf{RQ1}: Does the proposed framework reduce the number of experiments performed with an increase in the number of queries it answers?}

For the battery layout case-study, we repeatedly performed Pareto-optimal search in a pre-defined parameter space with 490000 potential configurations, defined as follows: 
\begin{align*}
\mathsf{Voltage} \in \{ 200 + n \cdot 0.1 \mid n \in \mathbb{N}_0,\ 200 + n \cdot 0.1 < 210 \} \wedge
&\\
\mathsf{MaxTorque} \in \{ 400 + n \cdot 0.1 \mid n \in \mathbb{N}_0,\ 400 + n \cdot 0.1 < 410 \} \wedge
&\\
\mathsf{InternalRes} \in \{ 0.02 + n \cdot 0.01 \mid n \in \mathbb{N}_0,\ 0.02 + n \cdot 0.01 \leq 0.50 \}
\end{align*}

Individual distance thresholds for fuzzy matching were used instead of specifying a single Euclidean distance threshold (as described above).
Three repetitions (starting from empty experiment stores) were carried out for each set of fuzzy matching distance thresholds in a decreasing order to simulate reduction in uncertainty as early engineering workflows progress over time. The search-space in the Pareto-optimization algorithm was explored randomly in each repetition. Four different distance measures were used. 

For the train system running example, starting from an empty experiment store, we repeatedly performed safety requests with random parameters (with 64-bit floating point precision) within the following space, and recorded the number of actually executed experiments at regular intervals:

\begin{align*}
    (100 \leq m \leq 20000 ) \wedge (0.02 \leq F_B \leq 2.5) \wedge (10 \leq v \leq 600 ) \\
    \wedge (0 \leq slip \leq 1 ) \wedge (-25 \leq \theta \leq 25 ) \wedge (200 \leq dist \leq 2000)
\end{align*}

\paragraph{\textbf{RQ2}: Does experiment reuse improve the time needed to answer a query?}

The evaluation was carried out on a workstation with an i7-10510U CPU, 16 GB RAM, and an NVMe SSD, by profiling the Python code in the prototype.

We profiled the impact of reuse on the time taken for Pareto-optimal search of the battery system layout parameters by repeatedly asking random optimization queries of 1000-point (10 equidistant points on each axis) subspaces within the configuration space defined by the following constraints:

\begin{align*}
\mathsf{Voltage} \in \{ 200 + n \cdot 10 \mid n \in \mathbb{N}_0,\ 200 + n \cdot 10 < 600 \} \wedge
&\\
\mathsf{MaxTorque} \in \{ 400 + n \cdot 10 \mid n \in \mathbb{N}_0,\ 400 + n \cdot 10 < 1000 \} \wedge
&\\
\mathsf{InternalRes} \in \{ 0.02 + n \cdot 0.01 \mid n \in \mathbb{N}_0,\ 0.02 + n \cdot 0.01 \leq 0.50 \}
\end{align*}
-- which resolves to 153,000 points. By the Dirichlet principle, it can be shown that there will be a guaranteed overlap in the subspaces of the queries, if more than $153,000 / 1000 = 153$ queries are asked, which will be useful to show the benefits of the reuse framework.

To evaluate the evolution of time taken in answering a train safety query, we repeatedly posed train safety querys with random parameters (within the previously defined range, and with 64-bit floating point precision), and measured the average time taken to answer the train safety queries as the number of queries increased. A number of different configurations (to deeply understand where time is being consumed) of experiment management systems in broadly two categories -- with and without the storage of experiment trace data, described as follows, were tested:

\begin{enumerate}
\item Storage of neither the experiment data nor the experiment metrics. 

\item Storage of the experiment metrics (which is required for reuse) with the proposed tripartite framework.

\item Storage of the experiment metrics without any reuse mechanism or inference. 

\item Storage of the experiment traces (for full reuse) with the proposed tripartite reuse framework.

\item Storage of the experiment traces without any reuse mechanism or inference. 
\end{enumerate}

\paragraph{\textbf{RQ3}: Does the framework support explicating the domain knowledge needed for reuse of experiments?}

We evaluated this based on a study of the prototypical artifacts (code) for the reuse mechanisms.

\subsubsection{Evaluation Results}

The results of the evaluations are described w.r.t each research question:

\paragraph{\textbf{RQ1}: Does the proposed framework reduce the number of experiments performed with an increase in the number of queries it answers?} 
Evaluations of the battery system case-study and the train braking running example show that the proposed framework indeed reduces the number of experiments performed.

The results from the battery system case-study are displayed in \Cref{tab:battery-table}. 

The results corroborate the intuition that as the maximum similarity distance for fuzzy reuse decreases (representing reducing uncertainty as early systems engineering progresses), the number of experiments actually performed increases. Only 24 experiments were executed on average (in a 490000 point space) for the highest distance threshold, defined by \[(t_\mathsf{Voltage} = 5, t_\mathsf{MaxTorque} = 5, t_\mathsf{InternalRes} = 0.05)\] In contrast, for the lowest distance threshold, an average of 1435 experiments were performed: \[(t_\mathsf{Voltage} = 0.5, t_\mathsf{MaxTorque} = 0.5, t_\mathsf{InternalRes} = 0.01)\]

Notably, the results also support the intuition that as the uncertainty decreases, the number of points in the Pareto-optimal space also decreases. While there were in average only 3587 Pareto-optimal configurations for the lowest distance threshold, there were 62192 average Pareto-optimal configurations for the highest distance threshold. 

The results from the evaluation on the train braking safety running example are displayed in \Cref{fig:train-ratio}. The figure displays the ratio of experiments actually performed versus the number of train configuration braking safety requests performed in succession. The results agree with the conclusion that there is indeed a reduction in the number of experiments performed (since the ratio is less than 1). Notice that the ratio seems to stabilize around a constant value ($\approx$ 0.2) as the number of requests increases. Repeated evaluations with different fuzzy reuse distance thresholds did not effect the value ($\approx$ 0.2), which suggests that this emergent ratio is dominated by the experiment parameter space and the case-based reasoning rules described in RQ3, rather than the fuzzy reuse distance threshold.

\begin{table}[]
\begin{adjustbox}{max width=\textwidth}
\begin{tabular}{|ccc|ccc|ccc|}
\hline
\multicolumn{3}{|c|}{Similarity Distance}                                & \multicolumn{3}{c|}{Experiments Performed}                                              & \multicolumn{3}{c|}{Points in Pareto Front}                           \\ \hline
\multicolumn{1}{|c|}{Voltage} & \multicolumn{1}{c|}{Resistance} & Torque & \multicolumn{1}{c|}{Round 1}       & \multicolumn{1}{c|}{Round 2}       & Round 3       & \multicolumn{1}{c|}{Round 1} & \multicolumn{1}{c|}{Round 2} & Round 3 \\ \hline
\multicolumn{1}{|c|}{0.5}     & \multicolumn{1}{c|}{0.01}       & 0.5    & \multicolumn{1}{c|}{\textbf{1387}} & \multicolumn{1}{c|}{\textbf{1569}} & \textbf{1348} & \multicolumn{1}{c|}{5295}    & \multicolumn{1}{c|}{3422}    & 2043    \\ \hline
\multicolumn{1}{|c|}{1}       & \multicolumn{1}{c|}{0.02}       & 1      & \multicolumn{1}{c|}{\textbf{430}}  & \multicolumn{1}{c|}{\textbf{322}}  & \textbf{297}  & \multicolumn{1}{c|}{9921}    & \multicolumn{1}{c|}{7624}    & 12544   \\ \hline
\multicolumn{1}{|c|}{2}       & \multicolumn{1}{c|}{0.03}       & 2      & \multicolumn{1}{c|}{\textbf{100}}  & \multicolumn{1}{c|}{\textbf{92}}   & \textbf{92}   & \multicolumn{1}{c|}{27415}   & \multicolumn{1}{c|}{40152}   & 32286   \\ \hline
\multicolumn{1}{|c|}{5}       & \multicolumn{1}{c|}{0.05}       & 5      & \multicolumn{1}{c|}{\textbf{28}}   & \multicolumn{1}{c|}{\textbf{24}}   & \textbf{19}   & \multicolumn{1}{c|}{58652}   & \multicolumn{1}{c|}{27525}   & 100400  \\ \hline
\end{tabular}
\end{adjustbox}
\caption{Evaluation of the battery system case-study -- a Pareto-optimal search of the best parameters to minimize $\mathsf{TBL}$ and maximize $\mathsf{SoC}$ in the parameter space: $
\mathsf{Voltage} \in \{ 200 + n \cdot 0.1 \mid n \in \mathbb{N}_0,\ 200 + n \cdot 0.1 < 210 \} 
$
$
\mathsf{MaxTorque} \in \{ 400 + n \cdot 0.1 \mid n \in \mathbb{N}_0,\ 400 + n \cdot 0.1 < 410 \} 
$
$
\mathsf{InternalRes} \in \{ 0.02 + n \cdot 0.01 \mid n \in \mathbb{N}_0,\ 0.02 + n \cdot 0.01 \leq 0.50 \}
$, which resolves to 490,000 points to be evaluated.}
\label{tab:battery-table}
\end{table}

\begin{figure}
    \centering
    \includegraphics[width=\linewidth]{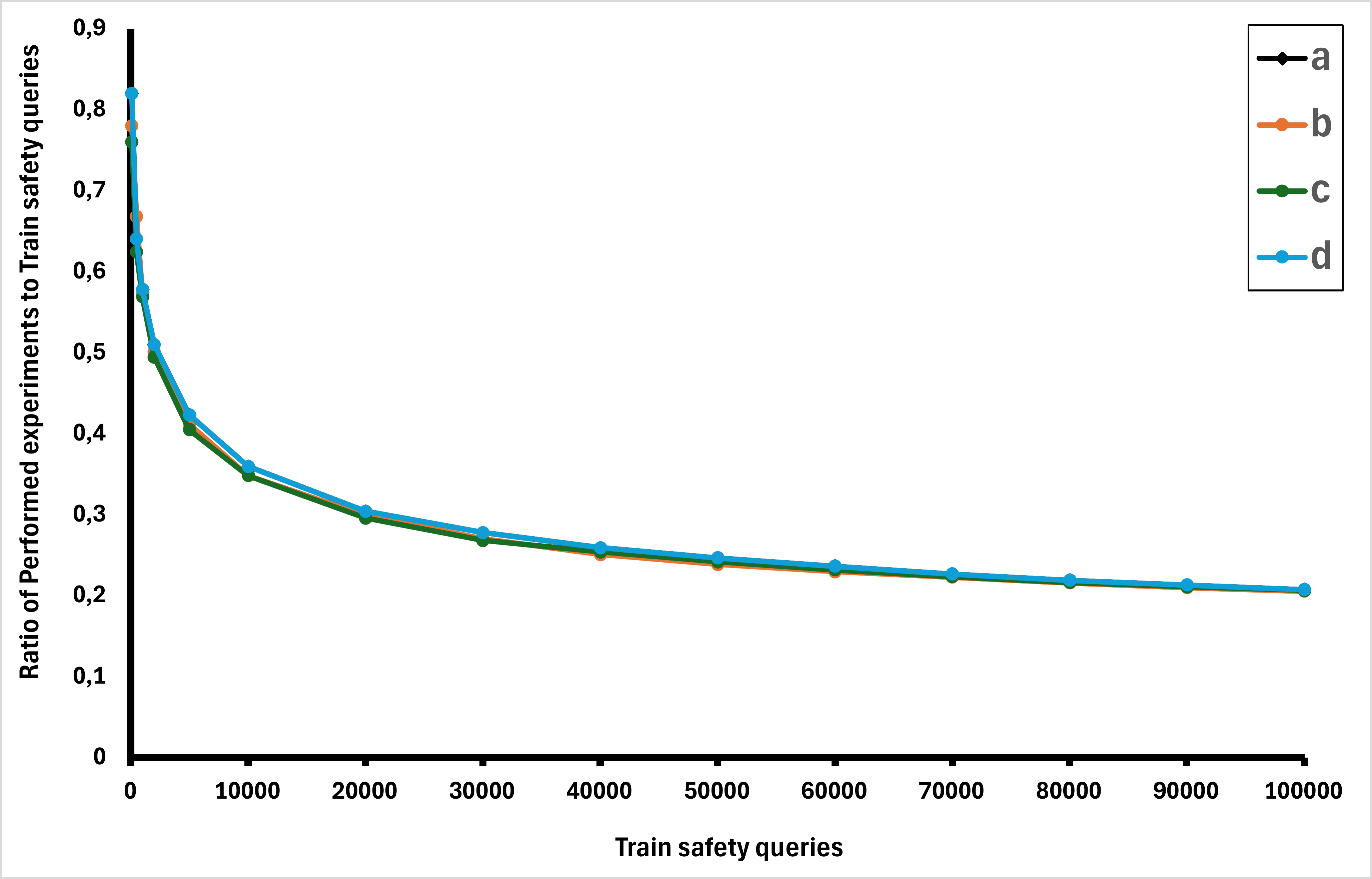}
    \caption{Evaluation of the running example (train braking), showing a decreasing ratio of experiments performed to number of queries as the number of queries grows, giving similar results with different fuzzy reuse distance thresholds:\\
    (a) $\{t_m \mapsto 100, t_{F_B}\mapsto 0.05, t_v\mapsto 0.5, t_\mathit{slip}\mapsto 0.05, t_\theta\mapsto 0.1\}$\\ 
    (b) $\{t_m \mapsto 200, t_{F_B}\mapsto 0.1, t_v\mapsto 1, t_\mathit{slip}\mapsto 0.1, t_\theta\mapsto 0.5\}$\\ 
    (c) $\{t_m \mapsto 500, t_{F_B}\mapsto 0.2, t_v\mapsto 2, t_\mathit{slip}\mapsto 0.2, t_\theta\mapsto 1\}$\\ 
    (d) $\{t_m \mapsto 1000, t_{F_B}\mapsto 0.5, t_v\mapsto 4, t_\mathit{slip}\mapsto 0.4, t_\theta\mapsto 2\}$}
    \label{fig:train-ratio}
\end{figure}

\paragraph{\textbf{RQ2}: Does experiment reuse improve the time needed to answer a query?} Based on evaluation of the prototypical implementation of the train braking safety example and battery layout case-study, the time needed to answer a query depends on the (a) nature of the experiment or simulation, (b) the quality and type of the experiment store, and (c) the level of reuse desired.

\Cref{fig:train-time-low} and \Cref{fig:train-time-high} describe the evolution of average time taken in the different configurations.

\begin{itemize}
\item With neither the experiment data nor the experiment metrics stored, the baseline approach with no experiment reuse took $\approx$ 16 milliseconds. 

\item If the experiment metrics are stored (which is required for reuse), the proposed tripartite framework was initially more efficient in answering train safety queries, but as the database size increased, queries started taking longer to process, and the average time increased, crossing the baseline at $\approx$ 50,000 queries.

\item If, for some reason, the experiment metrics are stored, but without any reuse mechanism or inference, the baseline approach took $\approx$ 28 milliseconds. From this, we can estimate that it took $\approx$ 12 milliseconds, to store the experiment metrics in the database. 

\item For full reuse, the experiment traces need to be stored as well. If the experiment traces are stored with the proposed tripartite reuse framework, it took $\approx$ 0.37 seconds to answer a safety query. A significant jump which can be attributed to the amount of time it takes to upload traces of the velocity and displacement of the train at a time-step of 0.01 seconds in each simulation (which is of variable time duration since the simulation is stopped 1 second in simulated time after the velocity reaches zero). However, the benefits of storing the traces, i.e. fuzzy recomputation of new experiment metrics is not explored in the evaluation and is not fully representative of the benefits.

\item If, for some reason, the experiment traces are stored, but without any reuse mechanism or inference, the baseline approach initially took the same amount of time as the tripartite, but rapidly increased as the number of queries increased, indicating a significant load on the database. 
\end{itemize}

In practical terms, the comparison between the baseline without storage and the tripartite framework with experiment metrics storage is the most directly relevant for assessing the potential efficiency gains from reuse. The results indicate that such studies can help in scheduling purges of the reuse-store in order to maximize the efficiency of the tripartite framework, along with faster database operations.

However, the other comparisons indicate that if an organization already stores the experimental data, then the tripartite framework can provide higher efficiency in answering queries. \Cref{fig:train-memory} describes the evolution of the size of the database in cases when all experiment data is stored for the baseline and tripartite approaches. The rate of increase in the size of the database is significantly lower for the tripartite framework compared to the baseline.

\begin{figure}
    \centering
    \includegraphics[width=\linewidth]{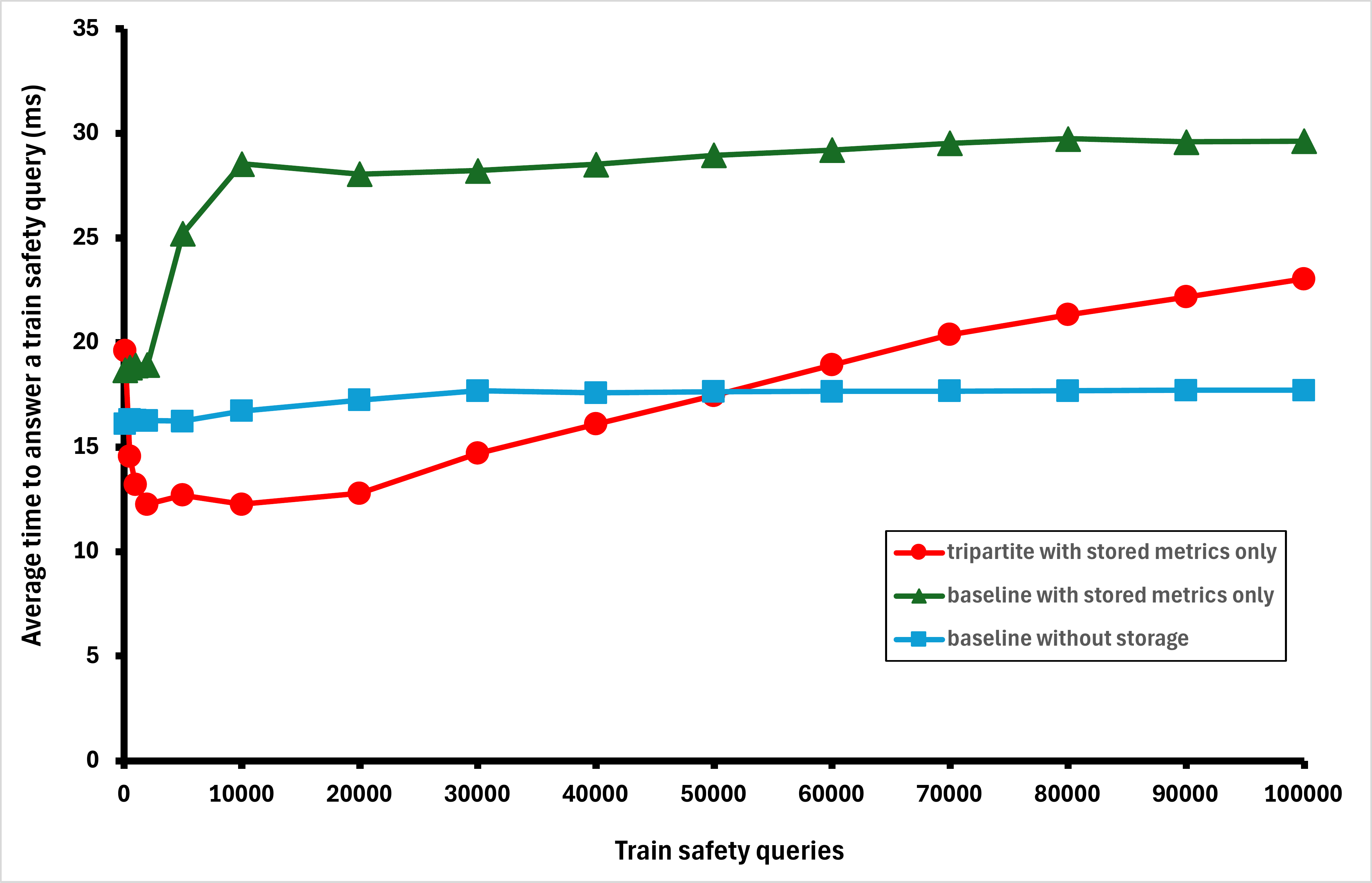}
    \caption{Evaluation of the running example (train braking), showing the evolution of the average time to answer a train safety query as the number of queries asked increases without storing experiment traces}
    \label{fig:train-time-low}
\end{figure}

\begin{figure}
    \centering
    \includegraphics[width=\linewidth]{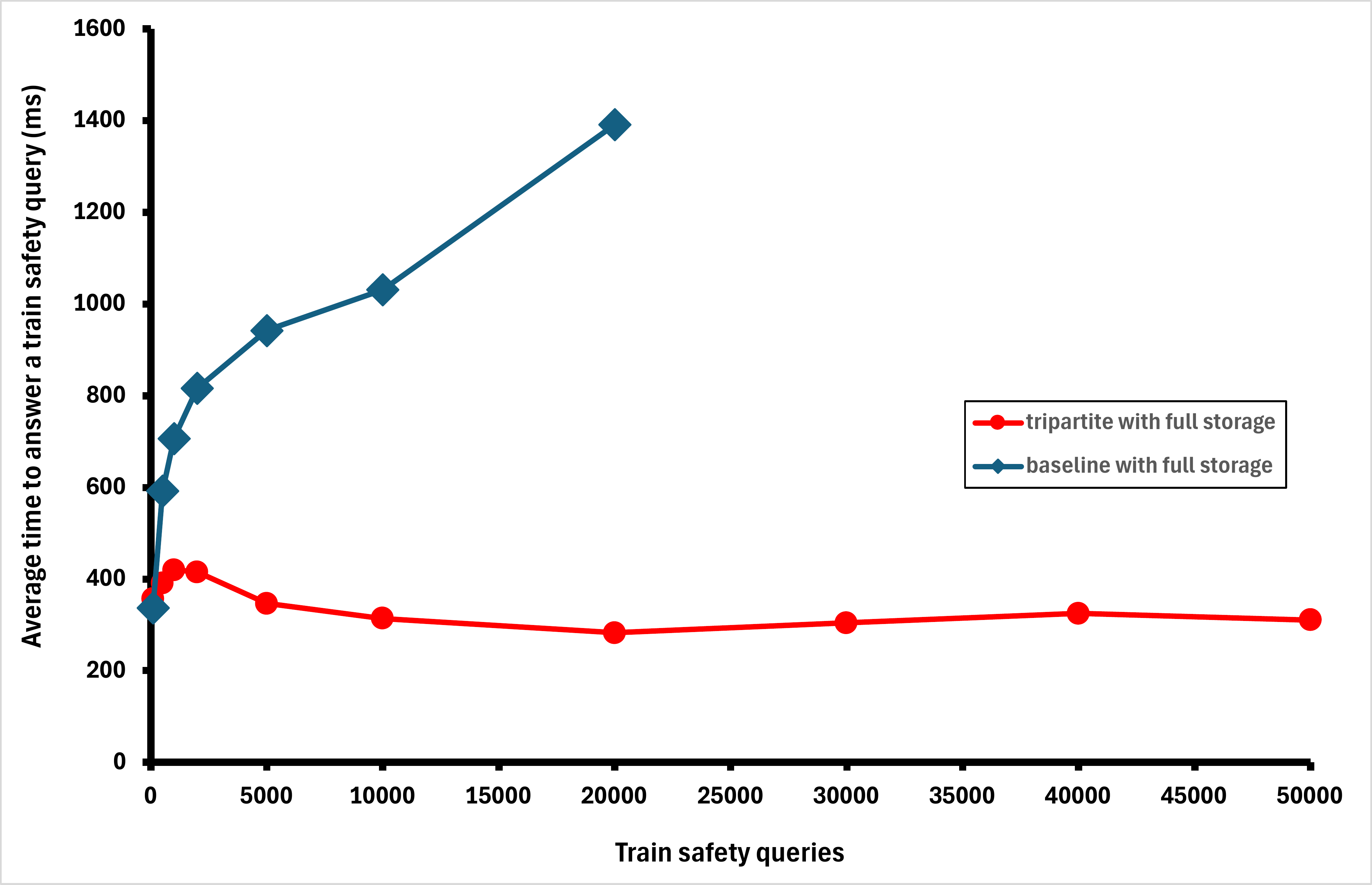}
    \caption{Evaluation of the running example (train braking), showing the evolution of the average time to answer a train safety query as the number of queries asked increases with storage of experiment traces}
    \label{fig:train-time-high}
\end{figure}

\begin{figure}
    \centering
    \includegraphics[width=\linewidth]{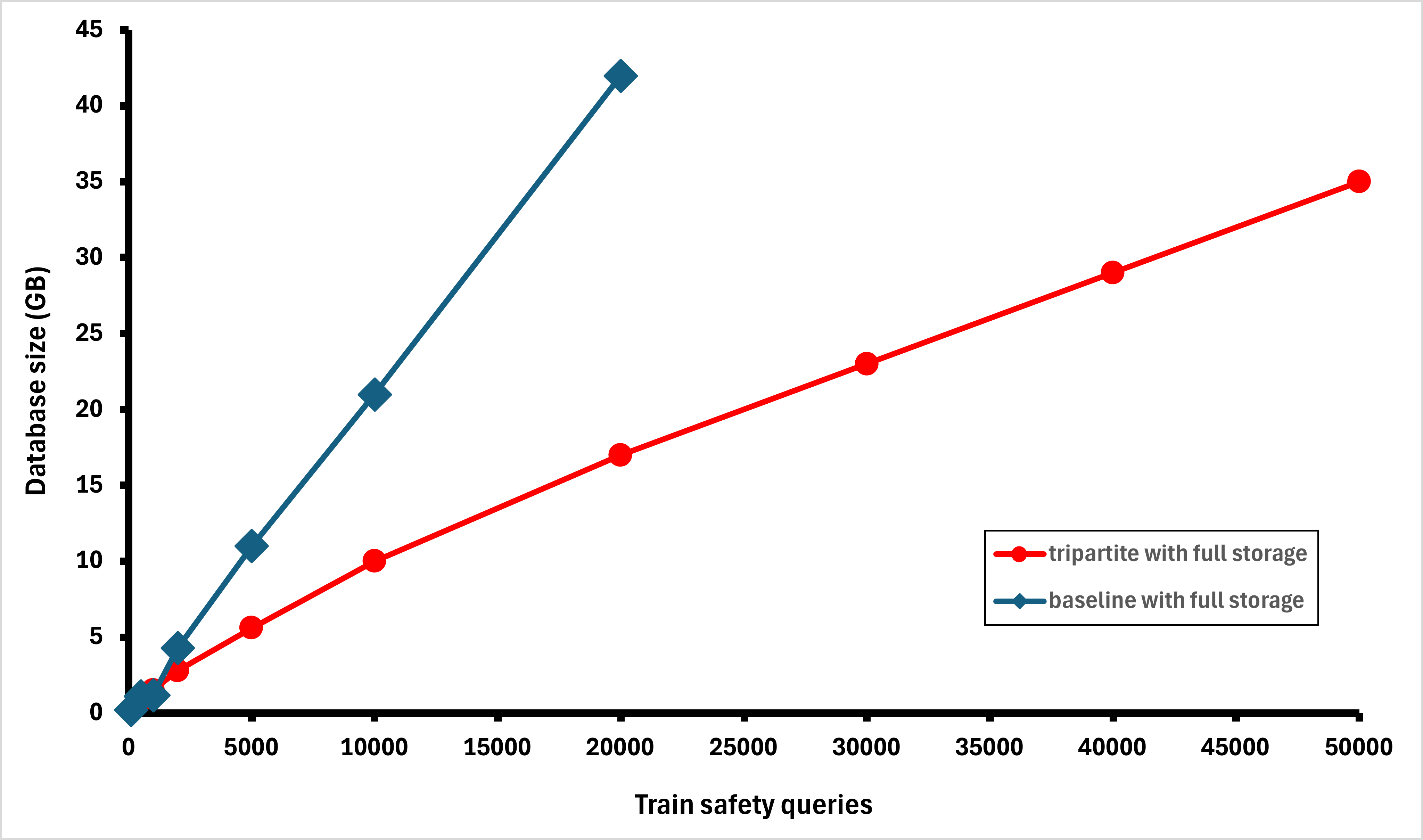}
    \caption{Evaluation of the running example (train braking), showing evolution of database size as more queries are asked, when both, experiment traces and metrics are stored.}
    \label{fig:train-memory}
\end{figure}

\Cref{fig:battery-time} describes the evolution of the average time taken to answer a battery optimization query. The graph uses dual-logarithmic scales (with base 2) to clearly show the profiling data. As can be seen, without reuse, it takes $\approx$ 558 seconds to answer a battery optimization query by performing an experiment for each of the 1000 points in each request. However, with the tripartite reuse framework, it takes an increasingly lower amount of time to answer a query, as the number of queries increases (reducing to and stabilizing to around 3 seconds by the 1000th query). 

\begin{figure}
    \centering
    \includegraphics[width=\linewidth]{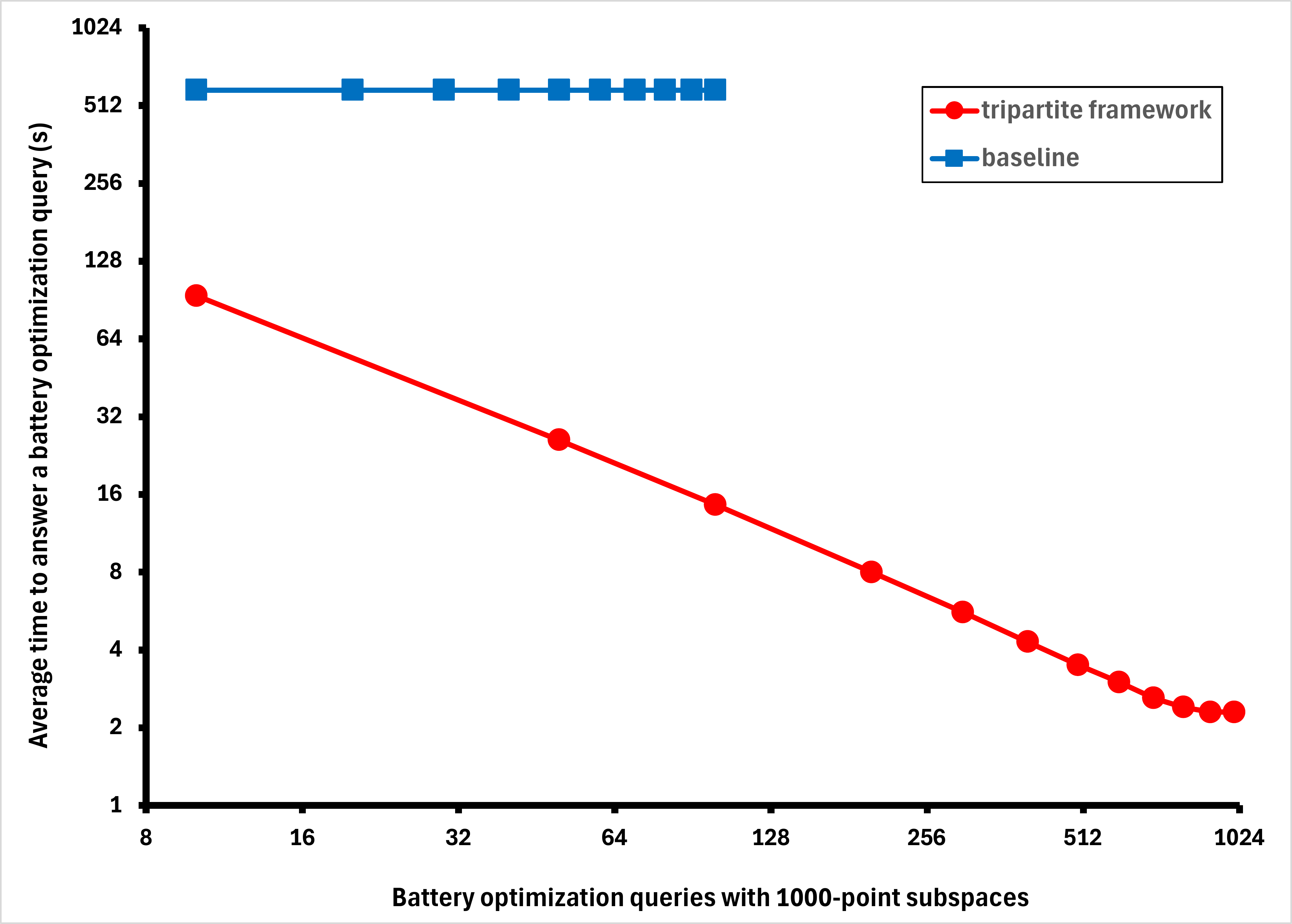}
    \caption{Evaluation of case-study (battery system), showing the evolution of average time taken to answer a query of Pareto-optimal analysis of the battery in a space of 1000 points, as the number of randomly-asked queries increases, on a dual-logarithmic scale (with base 2).}
    \label{fig:battery-time}
\end{figure}

\paragraph{\textbf{RQ3}: Does the framework support explicating the domain knowledge needed for reuse of experiments?} Based on an evaluation of the artifacts of the prototypical implementation, we draw the conclusion that the proposed framework does support explicating (and potential modeling) of the domain knowledge. 

\Cref{fig:is-inferable} describes in a simplified manner, the code for execution of the SQL database query used to infer if a battery configuration can be inferred as unstable based on pre-existing experiments in the experiment store. We can infer that having to explicate the necessary conditions for such a manner serves as grounds for (code) generation of the SQL query itself from a template which is configured based on a knowledge graph (that may possibly hold enough information to infer the inequality relations in the SQL query based on first principles).

For brevity, the code for the inferences in the train running example is described in \ref{sec:appendix-reuse-train-code}.

\begin{figure}
\begin{lstlisting}[language=pseudo, caption={Check if a battery configuration can be inferred as unstable based on pre-existing experiments requests in the SQL-based experiment store: reuse on the decomposition layer.}, label={fig:is-inferable}]
$\text{\textbf{function} }$is_inferable_unstable
$\text{\textbf{Input}:~~} \text{battery\_voltage, internal\_resistance, max\_torque}$
$\text{\textbf{Output}: \!\!\!} \text{Boolean (true if inferable as unstable, else false)}$

// Establish DB connection, and create a cursor to execute queries
conn := open_database_connection() 
cursor := conn.cursor()          

// Execute query to check for prior stable experiments
cursor.execute("""                
    SELECT 1
    FROM experiments
    WHERE
        system = `battery_system' AND
        (frame->>`battery_voltage')::numeric <= %s AND
        (frame->>`max_torque')::numeric <= %s AND
        (frame->>`internal_resistance')::numeric >= %s AND
        (metrics->>`soc')::numeric < 50;
""", (battery_voltage, max_torque, internal_resistance))

if cursor.fetchone() $\neq$ None return true 
return false                            
\end{lstlisting}
\end{figure}

\subsubsection{Discussion and Threats to Validity}

One threat to validity is the simplicity of the prototyped simulations with few configurable parameters. However, the battery case study is taken from an industrial context, so this threat is mitigated by the observation that low-parameter configuration spaces can occur, and benefit on experiment reuse if the underlying simulation is non-trivial. 

As described through the results of the research questions already, the efficiency of the proposed framework depends on the efficiency of its implementation, in particular, the underlying reasoner and databases, the nature of the experiments and the domain. This is mitigated by the use of off-the-shelf components as far as possible.

Our investigations only consider simulation experiments, and we do not perform real-world physical experiments, which even if configured and performed in an automated fashion are still subject to real-world temporal limitations, and traditionally have a higher cost. Our intuition is that the benefit in efficiency of the proposed framework is directly proportional to the average time it takes to perform an experiment.

The exploration of the experiment parameter search space in the Pareto-optimization of the battery was random. Capturing additional domain knowledge, e.g., about which queries are more likely or about typical sequences of queries, may allow to explicate the search strategy in order to maximize the yields in efficiency from case-based reasoning to reuse experiments. 
\section{Related Work}\label{sec:background}
\paragraph{Early Systems Engineering Decision Making} 

Several systematic studies have been performed on the notion of MBSE~\cite{henderson2021value,cederbladh2023early,ma2022systematic,carroll2016systematic,rashid2015toward}. A commonality between these studies is the argument that MBSE enables \textit{earlier} decision-making \& defect identification and by extension reduces costs during development. However, MBSE is highly multidisciplinary, resulting in a plethora of contextually viable instances of early V\&V. The ``best practice'' early decision-making is therefore not commonly defined~\cite{cederbladh2023early}. It is necessary to address the high degree of uncertainty at early stages~\cite{aroonvatanaporn2010reducing,driscoll2022decision}. One common approach is to re-use other project models directly, accepting differences in scope~\cite{lange2018systematic}. This provides a simple way to re-use existing artifacts, but is often applied ad hoc. Another common application area is the subsequent design of a single model acting as the single source of truth~\cite{abdo2022model,puntigam2020integrated}. While often leading to a greater reduction in uncertainty, this enforces a significant effort for each new development. 

Towards the future it is expected that various silos of an organization should have automatic means of data exchange and collaboration using models~\cite{cederbladh2024road}. In the case of industrial MBSE, this most often takes the form of simulation enabled as an early and continuous capability of development~\cite{puntigam2020integrated,nigischer2021multi,parnell2021mbse}. Kannan \cite{kannan2021formal} discusses how formal representations of individual knowledge can support reasoning in the decision-making process of large-scale complex engineering systems. In \cite{wakers}, the target is to capture the state of knowledge of individuals for formal reasoning to understand whether system designs are consistent with engineering standards and guidelines. 



\paragraph{Reuse in Model-Based Systems Engineering}

Systematic reuse is a cornerstone of effective MBSE, aiming to establish structured approaches for leveraging a wide array of engineering artifacts. The successful implementation of such reuse strategies is intrinsically linked to the diverse landscape of MBSE methodologies and the capabilities of supporting tool-chains. De Saqui-Sannes et al.~\cite{de2022taxonomy} provide a taxonomy that categorizes the varied MBSE approaches by languages, tools, and methods, highlighting the context in which reuse must operate. Ma et al.~\cite{ma2022systematic} perform a systematic literature review of MBSE tool-chains, which highlights the importance of tooling in enabling and managing reusable assets across the SE life-cycle. Systematic reuse extends beyond the mere ad-hoc utilization of existing models and includes requirements patterns, architectural designs, verification and validation (V\&V) scenarios, and workflows. The benefits of such systematic reuse are accelerated development cycles, reduced costs, enhanced system quality and consistency, effective preservation and dissemination of domain-specific knowledge, etc. 

Furthermore, ensuring the adaptability and applicability of these assets to new contexts, managing their variability, and maintaining the integrity and traceability of reused elements are critical concerns \cite{clements2001software}. Foundational concepts from Software Product Line Engineering (PLE) provide established methodologies for managing reuse across a family of systems, offering valuable paradigms for MBSE \cite{clements2001software}. More recent advancements include the application of ontologies and semantic technologies to improve the discoverability and interoperability of reusable knowledge \cite{GUO2024}. Additionally, emerging AI-driven techniques are being explored to assist in identifying, recommending, and even automating the integration of suitable reusable components \cite{ai-mbse-reuse}. Effective reuse strategies are necessary for supporting the early and continuous V\&V processes in MBSE. Indeed, our proposed framework is concerned specifically with the reuse of data associated with experiments, and not on the reuse of specific software components. However, we believe the approaches described above our complementary to our framework and possibly domain knowledge can be reused to infer the experiment reuse CBR rules for our framework.

\paragraph{Modeling and Managing (Models of) Experiments}
One of the earliest models of an experiment in the simulation community is the `experimental frame' \cite{ziegler-ef}, which is a specification of the conditions under which the system is observed or
experimented with. This concept materializes as a simulator that consists of coupled components, the Transducer, Acceptor, and Generator, that respectively, observe output from the system, analyze and validate the observation, and provide input to the system \cite{ef-arch}. Collections of `experiment frame's, called `validity frames', are reportedly used to explicate the validity of a model w.r.t the system it models \cite{AckerMVD24, mittal23}. A domain-specific language for simulation experiments, SESSL (Scala-Enabled Simulation Specification Layer), has been proposed to specify and instantiate simulations for a number of simulators using Scala \cite{sessl}. The ESS (EMF Simulation Specification) was used to model and instantiate simulations in Simulink using the Eclipse Modeling Framework \cite{ess}. While the above approaches describe how to model an experiment at a very low-level of abstraction, they do not describe the methodology of deriving such experiment specifications, starting from a basic high-level user question. We believe our approach is compatible and agnostic to the specific language used to model experiment specifications (as listed above).

The numerous models of different experiments need to be stored and effectively managed in order to systematically utilize them. Such tasks often involve maintaining consistency between artifacts, tracking versions, maintaining and updating links, etc. Model management is the set of systematic strategies deployed in order to accomplish the aforementioned tasks. Various model management techniques and tools have been proposed by the academia (Modelverse \cite{modelverse}, OpenFlexo \cite{openflexo}, ACMoM \cite{acmom}, DesignSpace \cite{designspace}, to name a few) and industry (Model Identity Card \cite{10.1145/3652620.3688223}, OpenMBEE \cite{openmbee}, OpenCAESAR \cite{opencaesar}, MagicDraw \footnote{\url{https://www.3ds.com/products/catia/no-magic}}, etc.). The ISO-14258 standard \cite{iso-14258} categorizes the various model management strategies into unification, integration, and federation, each with its own benefits and drawbacks. We consider the possible implementation of our proposed framework within a general model-management framework (among the ones described above) as potential future work.

\paragraph{Intelligent Caching}

Intelligent caching mechanisms are a strategy for optimizing performance and responsiveness within data-intensive environments, particularly those characterized by complex simulations, extensive analyses, and collaborative model development. The fundamental principle involves storing frequently accessed or computationally expensive artifacts—such as simulation results, validated model configurations, derived analytical data, or intermediate model transformation products—to obviate the need for redundant computations or protracted data retrievals. Traditional caching algorithms, such as Least Recently Used (LRU) or Least Frequently Used (LFU), rely on fixed heuristics and often perform suboptimally when faced with dynamic and complex access patterns prevalent in engineering workflows \cite{lru-lfu}. 

The ``intelligent'' aspect of modern caching leverages data-driven algorithms to create adaptive and predictive strategies. For instance, reinforcement learning has been employed to develop cache replacement policies that learn optimal decisions from interactions with the system environment, adapting to workload characteristics over time \cite{cache-replacement-ml}. Deep learning models, particularly Recurrent Neural Networks (RNNs) like LSTMs, have shown promise in predicting future data access patterns based on historical sequences, allowing for proactive caching decisions \cite{cache-replacement-deepl}. Some approaches, like LeCaR, focus on learning an optimal mix of fundamental policies (e.g., recency and frequency) based on their observed performance, effectively adapting the caching strategy to the current workload \cite{vietri2018driving}. We hinted in the evaluation of RQ2 that further evaluations are required in order to optimize the reuse of data from the cache, and keep the time to answer a train safety question low. In this manuscript, we have described what capabilities and features, cachs for experiment data need to have in order to enable reuse, however we leave particular cache optimization strategies as future work since it is our intuition that cache optimization approaches can be very domain- and problem-specific.

\section{Conclusion}\label{sec:conclusion}
Model-based approaches enable early experimenting for early V\&V, but early experiments can require significant cost and effort. We presented a formal framework for the tripartite decomposition of experiment-based information querying platforms, in order to leverage past experiment knowledge, combined with explicated domain knowledge, to reduce the number of experiments needed to perform for an answer within a reasonable threshold of accuracy. We presented our framework using a simplified physics-based running example of a train braking on slippery sloped tracks. Our formalization not only defines a rigorous collection of concepts and relations to rely on, but also establishes a common ground for future reuse architectures that intend to follow the same principles. We also prototyped our framework using the currently commonly available technologies. We evaluated the benefit of our framework via three research questions using an industrially relevant case-study of simulating a battery system FMU. The results corroborated our hypothesis and intuition, while also showing that the gains in efficiency are highly context- and implementation-dependent.

Although our evaluations provided promising results, we still need to explore our framework in different real-world scenarios and contexts. We discussed that our results are highly dependent on the implemented architecture. Exploring different architectures is also in our plans for future work. We only applied our strategy to simulation experiments. Even though we believe that for physical experiments the benefits may be even greater, we still need to provide evidence for that. We know that poor usability is a huge barrier to adoption. Thus, we also plan to investigate appropriate mechanisms for the user to define and execute experiment queries. Finally, in this work, we have exemplified different levels of reusability of experimental results. However, when these results cannot be reused, one possibility is to support the engineer in building experiments to answer the proposed query. The usage of domain knowledge allied to the architectures of previous experiments may contribute to that end. This is another research avenue in our plans. 

Therefore, in this paper, we have shown how knowledge of previously executed experiments may be used to save time and effort for engineers. We also outlined directions for possible research to further develop and strengthen the proposed framework. Engineers may relate some of their routines to part of the steps described in our approach. We believe that making these implicit good practices explicit in a systematic process can bring significant gains to system development processes.

\clearpage
\bibliographystyle{elsarticle-num}
\bibliography{main}
\clearpage
\appendix
\section{Reuse algorithms for the decomposition and execution layers}
\label{sec:appendix-reuse-layers-code}

\begin{figure}[hbp]
\begin{lstlisting}[language=pseudo, caption={Pseudo-code algorithm for reuse on the decomposition layer.}, label = {fig:algo-decomposition}]
$\text{\textbf{function} }$reuse-decomposition
$\text{\textbf{Input}:~~A request}$ $\request \text{ and a store}$ store
$\text{\textbf{Output}: \!\!\!A pair of a store and either a response, or } \bot$

//implements direct reuse
if($\mathsf{getResponse}($store$,\request) = \response$) return $\langle $store$, \response\rangle$

//implements symbolic reasoning
if($\exists \request' \mid \mathsf{getResponse}($store$,\request') = \response' \wedge \justify_R(\request,\request',\response') = \response$)
  return $\langle \mathsf{add}_R($store$,\request,\response), \response\rangle$
end

//implements fuzzy retrieval
if($\exists \request' \mid \mathsf{getResponse}($store$,\request') = \response \wedge |\request,\request'|^\mathsf{get}_\erlang < t^\mathsf{get}_\erlang$)
  return $\langle$store$, \response\rangle$
end

//implements fuzzy recomputation
if($\exists \request' \mid \mathsf{getResults}($store$,\request') \neq \bot \wedge~|\request,\request'|^\mathsf{comp}_\erlang < t^\mathsf{comp}_\erlang$)
  $\response$ := $\compute(\mathsf{getResults}($store$,\request'),\request)$
  return $\langle\mathsf{add}_R($store$,\request,\response), \response\rangle$
end
return $\langle$store$, \bot\rangle$
\end{lstlisting}
\end{figure}

\begin{figure}[hbp]
\begin{lstlisting}[language=pseudo, caption={Pseudo-code algorithm for reuse on the execution layer.}, label={fig:algo-execution}]
$\text{\textbf{function} }$reuse-execution
$\text{\textbf{Input}:~~An experiment specification}$ $\experiment\text{ and a store}$ store
$\text{\textbf{Output}: \!\!\!A pair of a store and either an answer, or } \bot$

//implements direct reuse
if($\mathsf{getResult}($store$,\experiment) = \exresult$) return $\langle $store$, \exresult\rangle$

//implements symbolic reasoning
if($\exists \experiment' \mid \mathsf{getResult}($store$,\experiment') = \exresult' \wedge \justify_E(\experiment,\experiment',\exresult') = \exresult$)
  return $\langle \mathsf{add}_E($store$,\experiment,\exresult), \exresult\rangle$
end

//implements fuzzy retrieval
if($\exists \experiment' \mid \mathsf{getResult}($store$,\experiment') = \exresult \wedge |\experiment,\experiment'|^\mathsf{get}_\exscheme < t^\mathsf{get}_\exscheme$)
  return $\langle$store$, \exresult\rangle$
end
return $\langle$store$, \bot\rangle$
\end{lstlisting}
\end{figure}

\section{Case-based reasoning algorithms for the train running example}
\label{sec:appendix-reuse-train-code}

\begin{figure}[hbp]
\begin{lstlisting}[language=pseudo, caption={Check if a train configuration can be inferred as unsafe based on pre-existing requests in the SQL-based store: reuse on the decomposition layer.}, label={fig:is-inferable-unsafe}]
$\text{\textbf{function} }$is_inferable_unsafe
$\text{\textbf{Input}:~~} \text{brake\_force, mass, init\_velocity, slope, slip\_factor, stop\_dist}$
$\text{\textbf{Output}: \!\!\!} \text{Boolean (true if inferable as unsafe, else false)}$

// Establish DB connection, and create a cursor to execute queries
conn := open_database_connection() 
cursor := conn.cursor()          

// Execute query to check for prior better experiments that were unsafe
cursor.execute("""                
    SELECT 1
    FROM experiments
    WHERE
        system = `train_system' AND
        (frame->>`brake_force')::numeric >= %s AND
        (frame->>`mass')::numeric <= %s AND
        (frame->>`initial_velocity')::numeric <= %s AND
        (frame->>`slope')::numeric >= %s AND
        (frame->>`slip_factor')::numeric <= %s AND
        (metrics->>`stopping_distance')::numeric > %s;
""", (brake_force, mass, init_velocity, slope, slip_factor, stop_dist))

if cursor.fetchone() $\neq$ None return true 
return false                            
\end{lstlisting}
\end{figure}

\begin{figure}[hbp]
\begin{lstlisting}[language=pseudo, caption={Check if a train configuration can be inferred as safe based on pre-existing requests in the SQL-based store: reuse on the decomposition layer.}, label={fig:is-inferable-safe}]
$\text{\textbf{function} }$is_inferable_safe
$\text{\textbf{Input}:~~} \text{brake\_force, mass, init\_velocity, slope, slip\_factor, stop\_dist}$
$\text{\textbf{Output}: \!\!\!} \text{Boolean (true if inferable as safe, else false)}$

// Establish DB connection, and create a cursor to execute queries
conn := open_database_connection() 
cursor := conn.cursor()          

// Execute query to check for prior worse experiments that were safe
cursor.execute("""                
    SELECT 1
    FROM experiments
    WHERE
        system = `train_system' AND
        (frame->>`brake_force')::numeric <= %s AND
        (frame->>`mass')::numeric >= %s AND
        (frame->>`initial_velocity')::numeric >= %s AND
        (frame->>`slope')::numeric <= %s AND
        (frame->>`slip_factor')::numeric >= %s AND
        (metrics->>`stopping_distance')::numeric < %s;
""", (brake_force, mass, init_velocity, slope, slip_factor, stop_dist))

if cursor.fetchone() $\neq$ None return true 
return false                            
\end{lstlisting}
\end{figure}
\clearpage
\tableofcontents

\end{document}